\documentclass[preprint]{ptephy_v1}

\usepackage{amsmath,graphics}

\newtheorem{definition}{Definition}

\newtheorem{lemma}{Lemma}
\newtheorem{corollary}{Corollary}
\newtheorem{proposition}{Proposition}

\newcommand{\be}{\begin{equation}}
\newcommand{\ee}{\end{equation}}
\newcommand{\vect}[1]{\boldsymbol{#1}}
\newcommand{\bbar}[1]{\mathbf{\bar{\mathnormal{#1}}}}
\renewcommand{\theenumi}{(\arabic{enumi})}

\begin{document}

\title{Physical and unphysical solutions
of random-phase approximation equation
}

\author{\name{H. Nakada}{\ast}}

\affil{Department of Physics, Graduate School of Science, Chiba University,\\
Yayoi-cho 1-33, Inage, Chiba 263-8522, Japan
\email{nakada@faculty.chiba-u.jp}}

\begin{abstract}%
Properties of solutions of the RPA equation is reanalyzed mathematically,
which is defined as a generalized eigenvalue problem
of the stability matrix $\mathsf{S}$
with the norm matrix $\mathsf{N}=\mathrm{diag.}(1,-1)$.
As well as physical solutions, unphysical solutions are examined in detail,
with taking the possibility of Jordan blocks of the matrix $\mathsf{N\,S}$
into consideration.
Two types of duality of eigenvectors and basis vectors of the Jordan blocks
are pointed out and explored,
which disclose many basic properties of the RPA solutions.
\end{abstract}

\subjectindex{RPA, duality, stability}

\maketitle

\section{Introduction}\label{sec:intro}

The mean-field (MF) theories provide us with good first approximation
of ground states in various quantum many-body systems.
For many-fermionic systems,
we cannot exaggerate importance of the Hartree-Fock (HF) theory,
which determines the one-body fields self-consistently
with respecting the variational principle~\cite{ref:RS80}.
The HF theory is appropriate for defining
single-particle (s.p.) orbitals of the constituent particles
in the individual system.
The density-functional theory (DFT)~\cite{ref:PY89,ref:ED11}
has mathematical and computational similarity to the HF theory,
in the points that the ground state properties are obtained
from the variational principle
with respect to the one-body fields,
although there are no physical meaning in the s.p. orbitals.
The MF theory has been extended
to the Hartree-Fock-Bogolyubov (HFB) theory~\cite{ref:RS80},
in which spontaneous breakdown of the particle-number conservation
(\textit{i.e.}, the global gauge symmetry) is taken into account
while respecting the self-consistency and the variational principle.
The HFB theory is useful when two particles correlate as a pair
and the system lies in the superconducting or the superfluid phase.

The random-phase approximation (RPA)~\cite{ref:PB52}
gives description of excited states,
excitation energies and transition strengths to be precise,
on top of the HF solution~\cite{ref:Thou60}.
It is well linked to the linear response theory,
and is obtained also as the small amplitude limit
of the time-dependent HF theory~\cite{ref:RS80}.
As in the similarity between the HF and the DFT,
the small amplitude limit of the time-dependent DFT~\cite{ref:ED11}
resembles the RPA.
Analogously to the RPA on the HF,
the quasiparticle-RPA is formulated on top of the HFB theory,
which can describe excited states of systems in superfluidity.

The RPA and their extensions are basically built
upon a solution in the corresponding MF theory.
As long as the MF state lies at a distinct minimum,
the RPA solutions can be regarded as excited states from the MF state.
However, the RPA solutions could be unphysical
if the MF state is actually at a saddle-point
in the vector space defined in its vicinity.
Whereas the RPA is extensively applied,
stability of the RPA solutions is not always clear.
In localized self-bound systems like atomic nuclei,
spontaneous symmetry breakdown (SSB) necessarily occurs
in the MF regime~\cite{ref:Thou60}.
An obvious example is the breaking of the translational symmetry,
and another well-known example is the breaking of the rotational symmetry
in deformed nuclei.
The SSB leads to a Nambu-Goldstone (NG) mode,
which comes out as a zero-energy solution in the RPA.
Structure of the vector space around the MF minimum
could be further complicated in the high-spin states~\cite{ref:SM84}
and along the path of the collective motion~\cite{ref:HNMM}.
In numerical calculations within the MF theories,
we often assume certain symmetry and ignore some degrees of freedom (d.o.f.)
to save computational resources.
Even if a minimum is obtained in a MF calculation,
it does not guarantee that it remains to be a minimum
for d.o.f. ignored in the calculation,
as in a spherical HF calculation for a quadrupolarly deformed nucleus.
Whereas stability of the RPA solutions in vicinity of the MF minimum
was investigated in Refs.~\cite{ref:Thou61,ref:TV62}
by Thouless \textit{et al.},
more general arguments with rigorous mathematical treatment are desired
for a variety of extensions of the RPA developed to date.

In this article,
I reanalyze properties of RPA solutions mathematically,
in terms of the linear algebra.
Unphysical solutions as well as physical solutions are examined
in some detail.
Although one may consider that unphysical solutions are just meaningless,
they are useful for understanding properties of the RPA equation
more profoundly.
As a result, they help us to comprehend
in what manner physical solutions and NG modes come out.
The present analysis will be of practical significance as well
in coding programs for numerical calculations in the RPA,
because one often has to prepare for various situations in numerical studies,
without knowing structure of the vector space around the MF state sufficiently.

\section{RPA equation}\label{sec:RPAeq}

The RPA equation is written as
\be\begin{split}
 \sum_\beta \big[A_{\alpha\beta}\,X^{(\nu)}_\beta
 + B_{\alpha\beta}\,Y^{(\nu)}_\beta\big] &= \omega_\nu\,X^{(\nu)}_\alpha\,,\\
 \sum_\beta \big[B_{\alpha\beta}^\ast\,X^{(\nu)}_\beta
 + A_{\alpha\beta}^\ast\,Y^{(\nu)}_\beta\big] &= -\omega_\nu\,Y^{(\nu)}_\beta\,,
\end{split}\label{eq:RPAeq-a}\ee
where $\alpha,\beta$ represent particle-hole bases on the HF solution
or two quasiparticle bases on the HFB solution.
The matrices $A$ and $B$ are obtained
from the residual interaction and the HF s.p. (or the HFB q.p.) energies,
and satisfy
\be A=A^\dagger\,,\quad B=B^T\,. \label{eq:sym-AB}\ee
In order to cope with a variety of physical situations,
I start discussions only from this structure of the RPA equation,
without any further assumptions.
The solution of Eq.~(\ref{eq:RPAeq-a}) is
comprised of $\omega_\nu$ and $(X^{(\nu)}, Y^{(\nu)})$,
which correspond to the energy and the wave function
of the $\nu$-th excited state.
It is imposed that $(X^{(\nu)}, Y^{(\nu)})$ obeys the \textit{normalization},
\be\begin{split}
\sum_\alpha \big[X^{(\nu)}_\alpha X^{(\nu')\ast}_\alpha
 - Y^{(\nu)}_\alpha Y^{(\nu')\ast}_\alpha\big] &= \delta_{\nu\nu'}\,,\\
\sum_\alpha \big[X^{(\nu)}_\alpha Y^{(\nu')}_\alpha
 - Y^{(\nu)}_\alpha X^{(\nu')}_\alpha\big] &= 0\,.
\end{split}\label{eq:norm-a}\ee
The dimension of the $A$ and $B$ matrices is denoted by $D$.
$D$ is finite in many practical calculations.
Although $D$ can be infinite in principle,
the arguments here will cover infinite $D$ as a limiting case.

Properties of the RPA solutions are better argued
in terms of the $2D\times 2D$ matrices,
\be
 \mathsf{S}:=\begin{pmatrix} A&B\\ B^\ast&A^\ast \end{pmatrix}\,,~
 \mathsf{N}:=\begin{pmatrix} 1&0\\ 0&-1 \end{pmatrix}\,,~
 \mathsf{\Sigma}_x := \begin{pmatrix} 0&1\\1&0 \end{pmatrix}\,.
\label{eq:matrices}\ee
For the matrices $\mathsf{N}$ and $\mathsf{\Sigma}_x$,
$\mathsf{N}^2=\mathsf{\Sigma}_x^2=\mathsf{1}$
and $\mathsf{\Sigma}_x\,\mathsf{N}+\mathsf{N}\,\mathsf{\Sigma}_x=\mathsf{0}$
hold.
The matrix $\mathsf{S}$ is nothing but
the curvature matrix at the HF (or HFB) solution,
and is known as the stability matrix.
With the matrices defined in Eq.~(\ref{eq:matrices}),
the RPA equation (\ref{eq:RPAeq-a}) is expressed as
\be \mathsf{S}\,\vect{x}_\nu=\omega_\nu\mathsf{N}\,\vect{x}_\nu\,;\quad
 \vect{x}_\nu=\begin{pmatrix} X^{(\nu)}\\ Y^{(\nu)}\end{pmatrix}\,,
\label{eq:RPAeq-b}\ee
and the normalization condition (\ref{eq:norm-a}) as
\be \vect{x}_\nu^\dagger\,\mathsf{N}\,\vect{x}_{\nu'}=\delta_{\nu\nu'}\,,~
 \vect{x}_\nu^T\mathsf{\Sigma}_x\,\mathsf{N}\,\vect{x}_{\nu'}=0\,.
\label{eq:norm-b}\ee
The number of the solutions satisfying (\ref{eq:norm-b})
matches the dimension $D$ of the $A,B$ matrix in ideal cases.
However, it is not obvious whether there exist $D$ solutions
that satisfy the normalization condition (\ref{eq:norm-b}).

\begin{definition}\label{def:solvable}
If there exist $D$ independent solutions of $(\omega_\nu, \vect{x}_\nu)$
that satisfy Eqs.~(\ref{eq:RPAeq-b}), (\ref{eq:norm-b}) and $\omega_\nu>0$,
the RPA equation is said \textbf{fully solvable}.
\end{definition}

\section{General properties of RPA solutions}\label{sec:general}

\subsection{Dualities and eigenvalues}\label{subsec:duality}

The RPA equation (\ref{eq:RPAeq-b}) is equivalent
to the eigenvalue problem of the matrix $\mathsf{N\,S}$.
An eigensolution is defined by a set of an eigenvalue
and an eigenvector $(\omega_\nu,\vect{x}_\nu)$.
At the same time, Eq.~(\ref{eq:RPAeq-b}) reads
the eigenvalue problem of $\mathsf{S\,N}$
for an eigensolution $(\omega_\nu,\mathsf{N}\,\vect{x}_\nu)$.
This duality is important to derive basic properties of the RPA solutions.
The relevant duality will be established later,
and is called \textit{LR-duality} in this article,
because it connects left and right basis vectors.

Another important ingredient is the symmetry of $\mathsf{S}$.
The structure of $\mathsf{S}$ given in Eq.~(\ref{eq:matrices})
with (\ref{eq:sym-AB}) is characterized by
\be \mathsf{S}=\mathsf{S}^\dagger\,,\quad
 \mathsf{\Sigma}_x\,\mathsf{S}^\ast\,\mathsf{\Sigma}_x = \mathsf{S}\,.
\label{eq:S-prop}\ee
It should be noted that, owing to the first equation of (\ref{eq:S-prop}),
an eigenvector of $\mathsf{S\,N}$ associated with an eigenvalue $\omega_\nu$
immediately gives a left eigenvector of $\mathsf{N\,S}$
corresponding to the eigenvalue $\omega_\nu^\ast$,
because $\mathsf{S\,N}\,\vect{y}_\nu = \omega_\nu\,\vect{y}_\nu$
is equivalent to $\vect{y}_\nu^\dagger\,\mathsf{N\,S}
= \omega_\nu^\ast\,\vect{y}_\nu^\dagger$.
\begin{proposition}\label{theor:sym-eigen}
If $\omega_\nu$ is an eigenvalue of $\mathsf{N\,S}$,
$-\omega_\nu$ is also an eigenvalue with equal degeneracy.
So is $\omega_\nu^\ast$.
\end{proposition}
\begin{proof}
Equation~(\ref{eq:RPAeq-b}) gives the secular equation, whose l.h.s. is
\be\begin{split}
 \det(\mathsf{S}-\omega\,\mathsf{N})
 &= \det\big([\mathsf{S}-\omega\,\mathsf{N}]^T\big)
 = \big[\det(\mathsf{S}-\omega^\ast\,\mathsf{N})\big]^\ast\\
 &= \det(\mathsf{\Sigma}_x\,\mathsf{S}^\ast\,\mathsf{\Sigma}_x
 -\omega\,\mathsf{N})
 = \det(\mathsf{S}^\ast
 -\omega\,\mathsf{\Sigma}_x\,\mathsf{N}\,\mathsf{\Sigma}_x)\\
 &= \det(\mathsf{S}^\ast+\omega\,\mathsf{N})
 = \big[\det(\mathsf{S}+\omega^\ast\,\mathsf{N})\big]^\ast\,,
\end{split}\label{eq:secular-lhs}\ee
because of Eq.~(\ref{eq:S-prop}).
Thus the equations $\det(\mathsf{S}-\omega\,\mathsf{N})=0$,
$\det(\mathsf{S}-\omega^\ast\,\mathsf{N})=0$
and $\det(\mathsf{S}+\omega^\ast\,\mathsf{N})=0$ are all equivalent,
ensuring correspondence of the eigenvalues $\omega$, $\omega^\ast$
and $-\omega^\ast$.
This leads to the eigenvalue $-\omega$ as well.
\end{proof}
\noindent
It is straightforward to show the following corollary
from the second equation of (\ref{eq:S-prop}).
\begin{corollary}\label{lem:sym-sigma}
If $(\omega_\nu, \vect{x}_\nu)$ is an eigensolution of Eq.~(\ref{eq:RPAeq-b}),
$(-\omega_\nu^\ast, \mathsf{\Sigma}_x\vect{x}_\nu^\ast)$ is also a solution.
\end{corollary}
\noindent
This relation indicates another kind of duality,
which is called \textit{UL-duality} to distinguish from the LR-duality,
since it is related to interchange (with taking complex conjugate)
of the upper and the lower components of the basis vectors.
With respect to the normalization of Eq.~(\ref{eq:norm-b}),
the following relation is obtained,
\be (\mathsf{\Sigma}_x\vect{x}_\nu^\ast)^\dagger\,\mathsf{N}\,
 (\mathsf{\Sigma}_x\vect{x}_{\nu'}^\ast)
 = -\vect{x}_{\nu'}^\dagger\,\mathsf{N}\,\vect{x}_\nu\,. \ee
The second equation of (\ref{eq:norm-b}) is interpreted as a relation
between $\mathsf{\Sigma}_x\vect{x}_\nu^\ast$ and $\vect{x}_{\nu'}$,
$(\mathsf{\Sigma}_x\vect{x}_\nu^\ast)^\dagger\,\mathsf{N}\,\vect{x}_{\nu'}=0$.

For an arbitrary vector
${\displaystyle\vect{x}=\begin{pmatrix} X\\ Y\end{pmatrix}}$,
we have $\vect{x}^\dagger\,\mathsf{N}\,\vect{x}
= X^\dagger\,X - Y^\dagger\,Y$.
In this respect the `norm' is not positive-definite.
\begin{definition}\label{def:normalizable}
A vector $\vect{x}$ is said \textbf{normalizable}
when $\vect{x}^\dagger\,\mathsf{N}\,\vect{x}>0$.
\end{definition}
\noindent
Indeed, if $\vect{x}^\dagger\,\mathsf{N}\,\vect{x}=r^2\,(>0)$,
$\vect{x}/r$ is normalized in the respect of Eq.~(\ref{eq:norm-b}).
As noted above,
it is not guaranteed that an eigenvector $\vect{x}_\nu$ of $\mathsf{N\,S}$
can be normalized.
The normalizability of an eigenvector $\vect{x}_\nu$
is a key to whether the RPA equation is solvable.

In arguments with respect to the RPA~\cite{ref:Thou61},
the `orthogonality' between two vectors $\vect{x}$ and $\vect{y}$
is sometimes defined when $\vect{x}^\dagger\,\mathsf{N}\,\vect{y}=0$,
by regarding $\mathsf{N}$ as the metric.
However, in this article I use the usual definition
that $\vect{x}$ and $\vect{y}$ are orthogonal
when $\vect{x}^\dagger\,\vect{y}=0$.

The hermiticity of $\mathsf{S}$ leads to the relation
between the solutions of Eq.~(\ref{eq:RPAeq-b}),
\be \vect{x}_\nu^\dagger\,\mathsf{S}\,\vect{x}_{\nu'}
 = \omega_\nu^\ast\,\vect{x}_\nu^\dagger\,\mathsf{N}\,\vect{x}_{\nu'}
 = \omega_{\nu'}\,\vect{x}_\nu^\dagger\,\mathsf{N}\,\vect{x}_{\nu'}\,.
\label{eq:norm-relation}\ee
Equation~(\ref{eq:norm-relation}) concludes:
\begin{lemma}\label{lem:bi-orthogonality}
For eigensolutions $(\omega_\nu, \vect{x}_\nu)$
and $(\omega_{\nu'}, \vect{x}_{\nu'})$ of the RPA equation (\ref{eq:RPAeq-b}),
$\omega_\nu^\ast=\omega_{\nu'}$
or $\vect{x}_\nu^\dagger\,\mathsf{N}\,\vect{x}_{\nu'}=0$ follows.
\end{lemma}
\noindent
The case of $\nu=\nu'$ was argued in Ref.~\cite{ref:Thou61}.

It is now possible to classify solutions of Eq.~(\ref{eq:RPAeq-b}).
\begin{proposition}\label{prop:eigenvalues}
Eigenvalues (denoted by $\omega_\nu$) of $\mathsf{N\,S}$ come out
in one of the following manners:
\begin{enumerate}
\item\label{item:positive}
 $\omega_\nu>0$ with a normalizable eigenvector $\vect{x}_\nu$,
 in association with another eigensolution
 $(-\omega_\nu, \mathsf{\Sigma}_x\vect{x}_\nu^\ast)$.
\item\label{item:negative}
 $\omega_\nu>0$ with an unnormalizable eigenvector,
 in association with an eigenvalue $-\omega_\nu$ that could be normalizable.
\item\label{item:pure-imag}
 a pair of pure imaginary eigenvalues, $\pm\omega_\nu$
 with $\mathrm{Re}(\omega_\nu)=0$, $\mathrm{Im}(\omega_\nu)\ne 0$.
\item\label{item:quartet}
 a quartet of complex eigenvalues, $\pm\omega_\nu, \pm\omega_\nu^\ast$
 with $\mathrm{Re}(\omega_\nu)\ne 0$, $\mathrm{Im}(\omega_\nu)\ne 0$.
\item\label{item:null} a null eigenvalue.
\end{enumerate}
\end{proposition}
\noindent
A solution belonging to Class~\ref{item:positive}
of Prop.~\ref{prop:eigenvalues} may be called \textit{physical solution},
while a solution to one of \ref{item:negative}\,--\,\ref{item:quartet}
\textit{unphysical solution}.
Class~\ref{item:null} is closely connected to the NG mode,
and a solution belonging to it will be called \textit{NG-mode solution}.
Focusing on solutions in vicinity of a MF minimum,
Thouless did not discuss
solutions of Classes~\ref{item:negative} and \ref{item:quartet}
in Refs.~\cite{ref:Thou61,ref:TV62}.

\subsection{Basis vectors in Jordan blocks}\label{subsec:Jordan}

Several eigenvalues of $\mathsf{N\,S}$ could be degenerate,
and degenerate eigenvalues may give rise to Jordan blocks.
This possibility is examined in this subsection.
Suppose that a Jordan block is generated from an eigenvector $\vect{x}_\nu$.
A basis vector of the Jordan block is denoted by $\vect{\xi}_k^{(\nu)}$,
which is obtained by
\be \mathsf{S}\,\vect{\xi}_{k+1}^{(\nu)}
 = \omega_\nu\,\mathsf{N}\,\vect{\xi}_{k+1}^{(\nu)}
 + ic_k^{(\nu)}\,\mathsf{N}\,\vect{\xi}_k^{(\nu)}\,.\quad
 (c_k^{(\nu)}\in\mathbf{C})
\label{eq:Jordan}\ee
The constant $c_k^{(\nu)}$ is subject to normalization and to relative phase
of $\vect{\xi}_k^{(\nu)}$ and $\vect{\xi}_{k+1}^{(\nu)}$.
Although it is taken to be $ic_k^{(\nu)}=1$ in the Jordan normal form,
another normalization will be adopted in Sec.~\ref{subsec:imag-sol}.
Starting from $\vect{\xi}_1^{(\nu)}=\vect{x}_\nu$,
$\vect{\xi}_{k+1}^{(\nu)}$ and $c_k^{(\nu)}\,(\ne 0)$ can be fixed
successively for $k\,(\geq 1)$.
This chain of equations ends at a certain $k$,
where no $\vect{\xi}_{k+1}^{(\nu)}\,(\ne\vect{0})$ exists.
Dimension of individual Jordan block is determined
by how long the chain continues, and is denoted by $d_\nu$.
Equation~(\ref{eq:Jordan}) defines a Jordan block of $\mathsf{S\,N}$ as well,
whose basis vectors are $\mathsf{N}\,\vect{\xi}_k^{(\nu)}$.

Most of the following arguments will cover the case of $d_\nu=1$;
\textit{i.e.}, the case that an eigenvector $\vect{x}_\nu$
does not generate a Jordan block.
\begin{lemma}\label{lem:Jordan-sym}
It there is a Jordan block for an eigenvalue $\omega_\nu$ with dimension $d_\nu$,
so are for $-\omega_\nu$ and $\omega_\nu^\ast$ with equal dimension $d_\nu$.
\end{lemma}
\begin{proof}
A Jordan block for the eigenvalue $\omega_\nu$ of $\mathsf{S\,N}$
directly corresponds to a Jordan block for $\omega_\nu^\ast$ of $\mathsf{N\,S}$,
by regarding the right basis vectors of $\mathsf{S\,N}$
as the left basis vectors of $\mathsf{N\,S}$.
Therefore the dimensions corresponding to $\omega_\nu$ and $\omega_\nu^\ast$
must be equal.\\
It follows from Eqs.~(\ref{eq:S-prop}) and (\ref{eq:Jordan}) that,
\be \mathsf{S}\,\mathsf{\Sigma}_x\vect{\xi}_{k+1}^{(\nu)\ast}
 = -\omega_\nu^\ast\,\mathsf{N}\,\mathsf{\Sigma}_x\vect{\xi}_{k+1}^{(\nu)\ast}
 + ic_k^{(\nu)\ast}\,\mathsf{N}\,\mathsf{\Sigma}_x\vect{\xi}_k^{(\nu)\ast}\,,
\label{eq:Jordan-sym}\ee
verifying that there exists a Jordan block for $-\omega_\nu^\ast$
with equal dimension.
Then, so is it for $-\omega_\nu$.
\end{proof}
\begin{definition}\label{def:ULdual}
The basis vectors $\vect{\xi}_k^{(\nu)}$
and $\mathsf{\Sigma}_x\vect{\xi}_k^{(\nu)\ast}$
are said to be \textbf{UL-dual} of each other.
\end{definition}
\noindent
If $\vect{\xi}_k^{(\nu)}=\mathsf{\Sigma}_x\vect{\xi}_k^{(\nu)\ast}$,
it is said to be \textit{self UL-dual}.

In general, a single eigenvalue $\omega_\nu$ may give plural Jordan blocks.
I hereafter reserve the subscript $\nu$
for representing individual Jordan block,
which is connected to a single eigenvector $\vect{x}_\nu$ of $\mathsf{N\,S}$,
rather than the eigenvalue.
The following proposition and lemmas are closely connected
to the structure of the Jordan block,
\be \begin{pmatrix}  \omega_\nu & ic_1^{(\nu)} & 0 & \cdots & 0 \\
  0 & \omega_\nu & ic_2^{(\nu)} & \ddots & 0 \\
  \vdots & \vdots & \ddots & \ddots & \vdots \\
  0 & 0 & 0 &  & ic_{d_\nu-1}^{(\nu)} \\
  0 & 0 & 0 & \cdots & \omega_\nu \end{pmatrix}\,,
\label{eq:Jordan-block}\ee
which is represented with the left and right basis vectors
that form the inverse matrix of each other apart from normalization.
Explicit proofs of Lemma~\ref{lem:Jordan-orthogonality} to
Prop.~\ref{theor:one-to-one_basis} are given
in Appendix~\ref{app:proof-Jordan},
which are instructive and useful
to confirm their compatibility with later propositions and lemmas.
\begin{lemma}\label{lem:Jordan-orthogonality}
Unless eigenvalues $\omega_\nu^\ast$ and $\omega_{\nu'}$ are equal,
any basis vector of $\mathsf{N\,S}$ associated with $\omega_\nu$
(either an eigenvector or a basis vector belonging to a Jordan block)
is orthogonal to any basis vector of $\mathsf{S\,N}$
associated with $\omega_{\nu'}$.
\end{lemma}

The above lemma suggests that
a Jordan block of $\mathsf{N\,S}$ for a specific $\omega_\nu$
and a Jordan block of $\mathsf{S\,N}$ for $\omega_\nu^\ast$ are paired.
\begin{lemma}\label{lem:Jordan-normalizability}
It is possible to take so that
each basis vector in a Jordan block of $\mathsf{N\,S}$
could overlap with no more than one basis vector
that contained in a Jordan block of $\mathsf{S\,N}$.
In order for the overlap not to vanish,
dimensions of these Jordan blocks of $\mathsf{N\,S}$ and $\mathsf{S\,N}$
must be equal.
\end{lemma}
\noindent
The next corollary follows
from the proof of Lemma~\ref{lem:Jordan-normalizability}
given in Appendix~\ref{app:sub-A2}:
\begin{corollary}\label{cor:non-Jordan}
If two eigenvectors of $\vect{x}_\nu$ and $\vect{x}_{\nu'}$ of $\mathsf{N\,S}$
satisfies $\vect{x}_\nu^\dagger\,\mathsf{N}\,\vect{x}_{\nu'}\ne 0$,
both of them do not constitute Jordan blocks (\textit{i.e.}, $d_\nu=d_{\nu'}=1$).
\end{corollary}
\begin{proposition}\label{theor:one-to-one_basis}
It is possible to produce a complete set of basis vectors of $\mathsf{N\,S}$,
by a proper transformation if necessary,
so that each of them could overlap with only one basis vector of $\mathsf{S\,N}$.
One-to-one correspondence is established
between Jordan blocks of $\mathsf{N\,S}$ and $\mathsf{S\,N}$
that contain basis vectors having non-vanishing overlaps.
The correspondence between basis vectors of $\mathsf{N\,S}$ and $\mathsf{S\,N}$
is also one to one.
\end{proposition}
\noindent
This proposition is an expression that the left and the right basis vectors
giving the Jordan representation constitute the inverse matrix of each other.

Proposition~\ref{theor:one-to-one_basis} enables us to define
a basis vector of $\mathsf{S\,N}$
that is dual of individual basis vector of $\mathsf{N\,S}$.
Let us denote the basis vector dual to $\vect{\xi}_k^{(\nu)}$
by $\bbar{\vect{\xi}}_k^{(\nu)}$.
Then $\bbar{\vect{\xi}}_k^{(\nu)\dagger}\,\mathsf{N}\,\vect{\xi}_{k'}^{(\nu')}$
does not vanish only for $\nu=\nu'$ and $k=k'$.
This is a realization of the Jordan block of (\ref{eq:Jordan-block}).
$\bbar{\vect{\xi}}_k^{(\nu)}$ is the $(d_\nu+1-k)$-th basis vector
of the Jordan block for an eigenvector
$\bbar{\vect{x}}_\nu:=\bbar{\vect{\xi}}_{d_\nu}^{(\nu)}$,
which is associated with the eigenvalue $\omega_\nu^\ast$.
Inverting Eq.~(\ref{eq:Jordan}), we obtain
\be \mathsf{S}\,\bbar{\vect{\xi}}_{k-1}^{(\nu)}
 = \omega_\nu^\ast\,\mathsf{N}\,\bbar{\vect{\xi}}_{k-1}^{(\nu)}
 - ic_{k-1}^{(\nu)\ast}\,
 \frac{\vect{\xi}_{k-1}^{(\nu)\dagger}\,\mathsf{N}\,\bbar{\vect{\xi}}_{k-1}^{(\nu)}}
 {\vect{\xi}_k^{(\nu)\dagger}\,\mathsf{N}\,\bbar{\vect{\xi}}_k^{(\nu)}}\,
 \mathsf{N}\,\bbar{\vect{\xi}}_k^{(\nu)}\,,
\label{eq:Jordan-inv}\ee
as confirmed in Appendix~\ref{app:sub-A4}.
This establishes the LR-duality of basis vectors.
\begin{definition}\label{def:LRdual}
The basis vectors $\vect{\xi}_k^{(\nu)}$ and $\bbar{\vect{\xi}}_k^{(\nu)}$
are said to be \textbf{LR-dual} of each other.
\end{definition}
\noindent
When $\vect{\xi}_k^{(\nu)}=\bbar{\vect{\xi}}_k^{(\nu)}$,
it is said \textit{self LR-dual}.
It is reasonable to define $\bbar{\vect{\xi}}_k^{(\nu)}$
so that it would fulfill
$\big|\bbar{\vect{\xi}}_k^{(\nu)\dagger}\,\mathsf{N}\,\vect{\xi}_{k'}^{(\nu')}\big|
=\delta_{\nu\nu'}\,\delta_{kk'}$.
It is not necessarily convenient to assume
$\bbar{\vect{\xi}}_k^{(\nu)\dagger}\,\mathsf{N}\,\vect{\xi}_{k'}^{(\nu')}
=\delta_{\nu\nu'}\,\delta_{kk'}$, though possible.
The normalizability in Def.~\ref{def:normalizable}
is now perceived as a part of the self LR-duality.

\subsection{Redefining eigenvalue problem}\label{subsec:redefine}

Because $\mathsf{S}$ is hermitian,
it is diagonalizable with an appropriate unitary matrix $\mathsf{U}$,
\be \mathsf{S}=\mathsf{U}^{-1}\,\mathrm{diag.}(\lambda_i)\,\mathsf{U}\,.\quad
 (\lambda_i\in\mathbf{R})
\label{eq:S-diag}\ee
By using this expression, $\mathsf{S}^{1/2}$ can be taken as
\be \mathsf{S}^{1/2}=\mathsf{U}^{-1}\,\mathrm{diag.}(\lambda_i^{1/2})\,\mathsf{U}\,.
\label{eq:S-half}\ee
Though $\lambda_i^{1/2}$ is two-valued, it is not important
which value is adopted.
A new matrix is now defined,
\be \tilde{\mathsf{S}}:=\mathsf{S}^{1/2}\,\mathsf{N\,S}^{1/2}\,.
\label{eq:S-tilde}\ee
Obviously, $\tilde{\mathsf{S}}$ is hermitian
only when $\lambda_i\geq 0$ for any $i\,(=1,\cdots,2D)$.

\begin{lemma}\label{lem:tilde-S_eigen}
All the eigenvalues and eigenvectors of $\mathsf{N\,S}$
correspond to those of $\tilde{\mathsf{S}}$, and vice versa.
\end{lemma}
\begin{proof}
If $\mathsf{S}\,\vect{x}_\nu=\omega_\nu\mathsf{N}\,\vect{x}_\nu$,
$\tilde{\mathsf{S}}\,(\mathsf{S}^{1/2}\vect{x}_\nu)
 =\omega_\nu(\mathsf{S}^{1/2}\vect{x}_\nu)$ follows.
Therefore $(\omega_\nu, \mathsf{S}^{1/2}\vect{x}_\nu)$
gives an eigensolution of $\tilde{\mathsf{S}}$
unless $\mathsf{S}^{1/2}\vect{x}_\nu=\vect{0}$.
Conversely, if $\tilde{\mathsf{S}}\,\vect{y}_\nu=\omega_\nu\vect{y}_\nu$,
$\mathsf{S}\,(\mathsf{N\,S}^{1/2}\vect{y}_\nu)
 =\omega_\nu\,\mathsf{N}\,(\mathsf{N\,S}^{1/2}\vect{y}_\nu)$ follows.\\
If $\mathsf{S}^{1/2}\vect{x}_\nu=\mathbf{0}$,
we have $\mathsf{S}\,\vect{x}_\nu=\tilde{\mathsf{S}}\,\vect{x}_\nu=0$,
indicating $\vect{x}_\nu$ is an eigenvector corresponding to the null eigenvalue
both of $\mathsf{N\,S}$ and $\tilde{\mathsf{S}}$.
\end{proof}

Thus the RPA equation is equivalent
to the eigenvalue problem of $\tilde{\mathsf{S}}$;
$\tilde{\mathsf{S}}\,\vect{x}_\nu=\omega_\nu\vect{x}_\nu$.
This redefinition is applied to investigate solvability of the RPA equation
in Sec.~\ref{subsec:real-sol}.

\section{Decomposition of vector space}\label{sec:decomp}

The whole vector space $\mathcal{V}$,
in which the stability matrix $\mathsf{S}$ is defined,
can be decomposed via the basis vectors produced by $\mathsf{N\,S}$
or those by $\mathsf{S\,N}$.
This furnishes further discussion on properties of the RPA solutions.

Recalling the LR-duality explored in Sec.~\ref{sec:general},
we obtain the projector
which separates out the direction along a certain basis vector
of $\mathsf{N\,S}$, as in Ref.~\cite{ref:TV62},
\be \mathsf{\Lambda}_{\nu,k} :=
 \frac{\vect{\xi}_k^{(\nu)}\,\bbar{\vect{\xi}}_k^{(\nu)\dagger}\,\mathsf{N}}
 {\bbar{\vect{\xi}}_k^{(\nu)\dagger}\,\mathsf{N}\,\vect{\xi}_k^{(\nu)}}\,.
\label{eq:proj-nuk}\ee
Obviously $\mathsf{\Lambda}_{\nu,k}\,\mathsf{\Lambda}_{\nu',k'}
=\delta_{\nu\nu'}\,\delta_{kk'}\,\mathsf{\Lambda}_{\nu,k}$.
The projector separating out the subspace
corresponding to the Jordan block (including the $d_\nu=1$ case)
generated from the eigenvector $\vect{x}_\nu$ is obtained by
\be \mathsf{\Lambda}_\nu := \sum_{k=1}^{d_\nu} \mathsf{\Lambda}_{\nu,k}\,. \ee
The projector $\mathsf{\Lambda}_\nu$ defines a subspace $\mathcal{W}_\nu$,
\be \mathcal{W}_\nu := \bigg\{ \sum_{k=1}^{d_\nu} a_k\,\vect{\xi}_k^{(\nu)};
 a_k\in\mathbf{C}\bigg\}\,, \label{eq:subspace-nu}\ee
for which $\mathsf{\Lambda}_\nu\mathcal{W}_\nu=\mathcal{W}_\nu$
and $(\mathsf{1}-\mathsf{\Lambda}_\nu)\mathcal{W}_\nu=\emptyset$ (empty set)
are satisfied.
The completeness is expressed as
\be \sum_\nu \mathsf{\Lambda}_\nu = \mathsf{1}\,,\quad
 \mathcal{V} = \bigoplus_\nu \mathcal{W}_\nu\,.\label{eq:proj-nu}\ee

In association with the LR-duality, one may consider
\be \bbar{\mathcal{W}}_\nu := \bigg\{ \sum_{k=1}^{d_\nu}
 a_k\,\bbar{\vect{\xi}}_k^{(\nu)}; a_k\in\mathbf{C}\bigg\}\,.
\label{eq:subspace-nubar}\ee
The projector corresponding to $\bbar{\mathcal{W}}_\nu$
is given by $\mathsf{N\,\Lambda}_\nu^\dagger\,\mathsf{N}$.
In order for the arguments in Sec.~\ref{sec:general} to be applicable
even after a certain projection,
it is desired to respect the UL-duality, as well as the LR-duality.
Therefore the UL-dual subspace is also considered,
\be \mathsf{\Sigma}_x\mathcal{W}_\nu^\ast := \bigg\{ \sum_{k=1}^{d_\nu}
 a_k\,\mathsf{\Sigma}_x\vect{\xi}_k^{(\nu)\ast}; a_k\in\mathbf{C}\bigg\}\,.
\label{eq:subspace-signu}\ee
The projector relevant to $\mathsf{\Sigma}_x\mathcal{W}_\nu^\ast$ is obtained by
\be \sum_{k=1}^{d_\nu} \frac{\mathsf{\Sigma}_x\vect{\xi}_k^{(\nu)\ast}\,
 \bbar{\vect{\xi}}_k^{(\nu)T}\mathsf{\Sigma}_x\,\mathsf{N}}
 {\bbar{\vect{\xi}}_k^{(\nu)T}\mathsf{\Sigma}_x\,\mathsf{N}\,
 \mathsf{\Sigma}_x\vect{\xi}_k^{(\nu)\ast}}
= \sum_{k=1}^{d_\nu} \frac{\mathsf{\Sigma}_x\vect{\xi}_k^{(\nu)\ast}\,
 \bbar{\vect{\xi}}_k^{(\nu)T}\,\mathsf{N}\,\mathsf{\Sigma}_x}
 {\bbar{\vect{\xi}}_k^{(\nu)T}\,\mathsf{N}\,\vect{\xi}_k^{(\nu)\ast}}
= \mathsf{\Sigma}_x\,\mathsf{\Lambda}_\nu^\ast\,\mathsf{\Sigma}_x\,. \ee
The subspace $\mathsf{\Sigma}_x\bbar{\mathcal{W}}_\nu^\ast$
and its relevant projector are defined as well.
Depending on $\omega_\nu$,
some of $\mathcal{W}_\nu$, $\bbar{\mathcal{W}}_\nu$,
$\mathsf{\Sigma}_x\mathcal{W}_\nu^\ast$
and $\mathsf{\Sigma}_x\bbar{\mathcal{W}}_\nu^\ast$ could be identical.
We shall use a collective index $[\nu]$
to stand for their direct sum,
\be \mathcal{W}_{[\nu]} := \mathcal{W}_\nu\oplus\bbar{\mathcal{W}}_\nu
 \oplus\mathsf{\Sigma}_x\mathcal{W}_\nu^\ast
 \oplus\mathsf{\Sigma}_x\bbar{\mathcal{W}}_\nu^\ast\,,\ee
apart from their overlap.
The projector $\mathsf{\Lambda}_{[\nu]}$ on $\mathcal{W}_{[\nu]}$
is defined by sum of the projectors.
Though in restricted cases,
a similar projection is considered in Ref.~\cite{ref:Don99}.
It is straightforward to show the following properties
of $\mathsf{\Lambda}_{[\nu]}$,
\be \mathsf{\Lambda}_{[\nu]}
 =\mathsf{N\,\Lambda}_{[\nu]}^\dagger\,\mathsf{N}\,,\quad
 \mathsf{\Lambda}_{[\nu]}\,\mathsf{\Lambda}_{[\nu']}
 =\delta_{[\nu],[\nu']}\,\mathsf{\Lambda}_{[\nu]}\,,\quad
 \mathsf{\Sigma}_x\,\mathsf{\Lambda}_{[\nu]}^\ast\,\mathsf{\Sigma}_x
 = \mathsf{\Lambda}_{[\nu]}\,.
\label{eq:proj-prop}\ee
\begin{lemma}\label{lem:dim-Lambda}
$d_{[\nu]}:=\dim \mathcal{W}_{[\nu]}$ must be even.
\end{lemma}
\begin{proof}
Unless $\mathcal{W}_\nu=\bbar{\mathcal{W}}_\nu
=\mathsf{\Sigma}_x\mathcal{W}_\nu^\ast
=\mathsf{\Sigma}_x\bbar{\mathcal{W}}_\nu^\ast$,
this is obvious from Lemma~\ref{lem:Jordan-sym}.
If $\mathcal{W}_\nu=\bbar{\mathcal{W}}_\nu
=\mathsf{\Sigma}_x\mathcal{W}_\nu^\ast
=\mathsf{\Sigma}_x\bbar{\mathcal{W}}_\nu^\ast$,
$\omega_\nu=0$ and $\bbar{\vect{\xi}}_k^{(\nu)}=\vect{\xi}_{d_\nu+1-k}^{(\nu)}$,
as is clear from the argument
with respect to Prop.~\ref{theor:one-to-one_basis}
in Appendix~\ref{app:sub-A3}.
Then, if $d_\nu$ is odd, we have
$\bbar{\vect{\xi}}_{(d_\nu+1)/2}^{(\nu)}=\vect{\xi}_{(d_\nu+1)/2}^{(\nu)}$,
contradictory to Lemma~\ref{lem:pure-imag-Jordan}
(Eq.~(\ref{eq:imag-norm}), to be more precise),
which leads to $\bbar{\vect{\xi}}_{(d_\nu+1)/2}^{(\nu)\dagger}\,\mathsf{N}\,
\vect{\xi}_{(d_\nu+1)/2}^{(\nu)}=0$.
See also the arguments in Appendix~\ref{app:sub-B1}.
\end{proof}

For an arbitrary $2D\times 2D$ matrix $\mathsf{M}$,
its projection onto the subspace $\mathcal{W}_{[\nu]}$
is obtained by
\be \mathsf{M}_{[\nu]} := \mathsf{\Lambda}_{[\nu]}^\dagger\,\mathsf{M}\,
 \mathsf{\Lambda}_{[\nu]}\,. \ee
Of next interest is how the projection affects
the RPA equation and the dualities.
\begin{proposition}\label{theor:SN-proj}
The projections of $\mathsf{S}$ and $\mathsf{N}$ satisfy
\be \mathsf{\Lambda}_{[\nu]}^\dagger\,\mathsf{S}\,\mathsf{\Lambda}_{[\nu']}
 = \delta_{[\nu],[\nu']}\,\mathsf{S}_{[\nu]}\,,\quad
 \mathsf{\Lambda}_{[\nu]}^\dagger\,\mathsf{N}\,\mathsf{\Lambda}_{[\nu']}
 = \delta_{[\nu],[\nu']}\,\mathsf{N}_{[\nu]}\,. \nonumber\ee
\end{proposition}
\begin{proof}
From Eq.~(\ref{eq:proj-prop}),
\be
 \mathsf{\Lambda}_{[\nu]}^\dagger\,\mathsf{N}\,\mathsf{\Lambda}_{[\nu']}
 = \mathsf{N}\,\mathsf{\Lambda}_{[\nu]}\,\mathsf{\Lambda}_{[\nu']}
 = \delta_{[\nu],[\nu']}\,\mathsf{N}\,\mathsf{\Lambda}_{[\nu]}
 = \delta_{[\nu],[\nu']}\,\mathsf{\Lambda}_{[\nu]}^\dagger\,\mathsf{N}\,
 \mathsf{\Lambda}_{[\nu]}\,.\ee
Concerning $\mathsf{S}_{[\nu]}$, Eq.~(\ref{eq:Jordan}) yields
\be \bbar{\vect{\xi}}_k^{(\nu)\dagger}\,\mathsf{S}\,\vect{\xi}_{k'}^{(\nu')}
 = \delta_{\nu\nu'}\,\big[\omega_\nu\,\delta_{kk'}
 + ic_k^{(\nu)}\,\delta_{k+1,k'}\big]\,
 \bbar{\vect{\xi}}_k^{(\nu)\dagger}\,\mathsf{N}\,\vect{\xi}_k^{(\nu)}\,,
\label{eq:Jordan-me}\ee
and therefore
\be\begin{split}
 \mathsf{\Lambda}_{[\nu]}^\dagger\,\mathsf{S}\,\mathsf{\Lambda}_{[\nu']}
 &= \mathsf{N}\,\mathsf{\Lambda}_{[\nu]}\,\mathsf{N\,S}\,\mathsf{\Lambda}_{[\nu']}
 = \mathsf{N}\,\bigg[\sum_{\nu\in[\nu]}\sum_{k=1}^{d_\nu}
 \frac{\vect{\xi}_k^{(\nu)}\,\bbar{\vect{\xi}}_k^{(\nu)\dagger}\,\mathsf{N}}
 {\bbar{\vect{\xi}}_k^{(\nu)\dagger}\,\mathsf{N}\,\vect{\xi}_k^{(\nu)}}\bigg]\,
 \mathsf{N\,S}\,\bigg[\sum_{\nu'\in[\nu']}\sum_{k'=1}^{d_{\nu'}}
 \frac{\vect{\xi}_{k'}^{(\nu')}\,\bbar{\vect{\xi}}_{k'}^{(\nu')\dagger}\,\mathsf{N}}
 {\bbar{\vect{\xi}}_{k'}^{(\nu')\dagger}\,\mathsf{N}\,\vect{\xi}_{k'}^{(\nu')}}\bigg]\\
 &= \delta_{[\nu],[\nu']}\,\mathsf{N}\,\sum_{\nu\in[\nu]}\bigg[
 \omega_\nu\sum_{k=1}^{d_\nu}\mathsf{\Lambda}_{\nu,k}
 +i\sum_{k=1}^{d_\nu}c_k^{(\nu)}\,
 \frac{\vect{\xi}_k^{(\nu)}\,\bbar{\vect{\xi}}_{k+1}^{(\nu)\dagger}\,\mathsf{N}}
 {\bbar{\vect{\xi}}_{k+1}^{(\nu)\dagger}\,\mathsf{N}\,\vect{\xi}_{k+1}^{(\nu)}}\bigg]
 = \delta_{[\nu],[\nu']}\,\mathsf{S}_{[\nu]}\,.
\end{split}\label{eq:Snu-separation}\ee
The last equality follows
because the quantity summed over $\nu$ does not depend on $[\nu']$.
\end{proof}
\begin{lemma}\label{lem:proj-sym}
$\mathsf{S}_{[\nu]}$ inherits the symmetry properties of Eq.~(\ref{eq:S-prop}),
\be \mathsf{S}_{[\nu]}=\mathsf{S}_{[\nu]}^\dagger\,,\quad
\mathsf{\Sigma}_x\,\mathsf{S}_{[\nu]}^\ast\,\mathsf{\Sigma}_x=\mathsf{S}_{[\nu]}\,.
\nonumber\ee
\end{lemma}
\begin{proof}
The hermiticity of $\mathsf{S}_{[\nu]}$ is obvious from its definition.
The second equation is proven as
\be \mathsf{\Sigma}_x\,\mathsf{S}_{[\nu]}^\ast\,\mathsf{\Sigma}_x
 = \mathsf{\Sigma}_x\,\mathsf{\Lambda}_{[\nu]}^T\,\mathsf{S}^\ast\,
 \mathsf{\Lambda}_{[\nu]}^\ast\,\mathsf{\Sigma}_x
 = \mathsf{\Lambda}_{[\nu]}^\dagger\,\mathsf{\Sigma}_x\,\mathsf{S}^\ast\,
 \mathsf{\Sigma}_x\,\mathsf{\Lambda}_{[\nu]}
 = \mathsf{\Lambda}_{[\nu]}^\dagger\,\mathsf{S}\,\mathsf{\Lambda}_{[\nu]}
 = \mathsf{S}_{[\nu]}\,,
\ee
from Eq.~(\ref{eq:proj-prop}).
\end{proof}
\noindent
Owing to Prop.~\ref{theor:SN-proj} and Lemma~\ref{lem:proj-sym},
$\mathsf{S}_{[\nu]}$ defines the RPA equation within $\mathcal{W}_{[\nu]}$
with maintaining both the LR- and UL-dualities\footnote{
A $d_{[\nu]}$-dimensional representation of $\mathsf{M}_{[\nu]}$
is provided by the matrix element
$\boldsymbol{\xi}_k^{(\nu)\dagger}\,\mathsf{N}\,\boldsymbol{\xi}_{k'}^{(\nu')}
/ \sqrt{|(\bar{\boldsymbol{\xi}}_k^{(\nu)\dagger}\,\mathsf{N}\,
\boldsymbol{\xi}_k^{(\nu)})\,(\bar{\boldsymbol{\xi}}_{k'}^{(\nu')\dagger}\,
\mathsf{N}\,\boldsymbol{\xi}_{k'}^{(\nu')})|}$.
}.
Thus the RPA equation is decomposed to the equation in each subspace $[\nu]$.
Furthermore, Prop.~\ref{theor:SN-proj} ensures
that the same holds for a direct sum of subspaces
$\mathcal{W}_{[\nu]}\oplus\mathcal{W}_{[\nu']}$,
via $(\mathsf{\Lambda}_{[\nu]}+\mathsf{\Lambda}_{[\nu']})^\dagger\,
\mathsf{S}\,(\mathsf{\Lambda}_{[\nu]}+\mathsf{\Lambda}_{[\nu']})
=\mathsf{S}_{[\nu]+[\nu']}$, and so forth.

The whole space $\mathcal{V}$ is thus decomposed
into the direct sum of $\mathcal{W}_{[\nu]}$,
\be \mathcal{V} = \bigoplus_{[\nu]} \mathcal{W}_{[\nu]}\,.\ee
Let us denote the complementary space of $\mathcal{W}_{[\nu]}$
by $\mathcal{W}_{[\nu]^{-1}}$
and the relevant projector by $\mathsf{\Lambda}_{[\nu]^{-1}}$;
\be \mathcal{W}_{[\nu]^{-1}} := \bigoplus_{[\nu']\,(\ne[\nu])} \mathcal{W}_{[\nu']}\,,\quad
 \mathsf{\Lambda}_{[\nu]^{-1}} := \mathsf{1}-\mathsf{\Lambda}_{[\nu]}
 = \sum_{[\nu']\,(\ne[\nu])} \mathsf{\Lambda}_{[\nu']}\,.
\label{eq:proj-comple}\ee
Analogously to $\mathsf{S}_{[\nu]}$,
$\mathsf{S}_{[\nu]^{-1}}:=\mathsf{\Lambda}_{[\nu]^{-1}}^\dagger\,\mathsf{S}\,
\mathsf{\Lambda}_{[\nu]^{-1}}$ defines the RPA equation in $\mathcal{W}_{[\nu]^{-1}}$,
keeping the LR- and UL-dualities
and eliminating the solutions within $\mathcal{W}_{[\nu]}$,
yet without influencing the solutions in $\mathcal{W}_{[\nu]^{-1}}$.
Therefore, all the disclosed properties of the RPA solutions in $\mathcal{V}$
are transferred to the solutions of the RPA equation
within $\mathcal{W}_{[\nu]^{-1}}$.
It is noted that the determinant of $\mathsf{S}$ is decomposed as well,
\be \det\,\mathsf{S} = \prod_{[\nu]} (\det\,\mathsf{S}_{[\nu]})\,,
\label{eq:detS-decomp}\ee
where $(\det\,\mathsf{S}_{[\nu]})$ on the r.h.s.
is the $d_{[\nu]}$-dimensional determinant.
The l.h.s of the secular equation, $\det(\mathsf{S}-\omega\,\mathsf{N})$,
can be expressed in an analogous manner.

\section{Properties of each class of solutions}
\label{sec:stable&unstable}

Properties of the RPA solutions are further analyzed
for individual class of Prop.~\ref{prop:eigenvalues}.
It deserves commenting here on degeneracy.
To the author's best knowledge,
possibility of degeneracy, particularly of Jordan blocks,
has not been examined well, except several specific NG modes.
Although degeneracy occurs even in physical solutions
under presence of certain symmetry
(\textit{e.g.}, degeneracy with respect to magnetic quantum numbers
under the rotational symmetry),
it does not give rise to Jordan blocks.
This is obvious when the conservation law allows us
to separate the RPA equation into the equations
according to the quantum numbers.
However, it is not trivial whether the same holds
for a variety of extensive applications of the RPA.
For instance, energy levels are highly degenerate in continuum,
as in the continuum RPA~\cite{ref:SB75}.
Consideration of the degeneracy could be relevant to
how we can take the continuous limit from arguments on discrete levels.
For the NG mode, Thouless restricted himself
to the case of two-dimensional Jordan blocks.
While higher-dimensional blocks are not very likely to emerge
in physical situations,
it will be meaningful to distinguish physical situations
from facts with rigorous mathematical proof.

\subsection{Solutions for real eigenvalues}\label{subsec:real-sol}

Let us first consider Classes~\ref{item:positive} and \ref{item:negative}
of Prop.~\ref{prop:eigenvalues}.

\begin{proposition}\label{theor:stability}
If the stability matrix $\mathsf{S}$ is positive-definite,
the RPA equation is fully solvable.
If the RPA equation is fully solvable, $\mathsf{S}$ is positive-definite.
\end{proposition}
\noindent
Although the first part of this proposition was already proven
in Ref.~\cite{ref:Thou61},
I prove it again in combination with the second part.
\begin{proof}
Suppose that $\mathsf{S}$ is positive-definite.
Then $\tilde{\mathsf{S}}$ in Sec.~\ref{subsec:redefine}
is hermitian and therefore diagonalizable
by a certain matrix $\tilde{\mathsf{X}}$,
$\tilde{\mathsf{S}}\,\tilde{\mathsf{X}}
=\tilde{\mathsf{X}}\,\tilde{\mathsf{\Omega}}$,
where $\tilde{\mathsf{\Omega}}$ is a diagonal matrix.
$\tilde{\mathsf{X}}$ can be unitary, but we shall take another normalization.
The eigenvalues in $\tilde{\mathsf{\Omega}}$ are all real and non-zero,
since $\det\tilde{\mathsf{S}}=\det(\mathsf{N\,S})\ne 0$.
Proposition~\ref{theor:sym-eigen} tells us
that they are pairwise, $\pm\omega_\nu$ ($\nu=1,\cdots,D$).
We here take ${\displaystyle\tilde{\mathsf{\Omega}}:=\begin{pmatrix}
 \mathrm{diag.}(\omega_\nu)&0\\ 0&-\mathrm{diag.}(\omega_\nu) \end{pmatrix}}$
so that $\omega_\nu>0$,
and define $\mathsf{\Omega}:=\mathsf{N}\,\tilde{\mathsf{\Omega}}$,
which is diagonal and positive-definite.
Let us adopt the normalization of $\tilde{\mathsf{X}}$
as $\tilde{\mathsf{X}}^\dagger\,\tilde{\mathsf{X}}=\mathsf{\Omega}$
(\textit{i.e.}, $\tilde{\mathsf{X}}\,\mathsf{\Omega}^{-1/2}$ is unitary)
and define $\mathsf{X}:=\mathsf{S}^{-1/2}\,\tilde{\mathsf{X}}$.
This derives
\be \mathsf{S\,X}=\mathsf{N\,X\,N\,\Omega}\,,\quad
 \mathsf{X}^\dagger\,\mathsf{N\,X}=\mathsf{N}\,, \label{eq:RPAeq-c}\ee
proving that the RPA equation is fully solvable because,
if we write $\mathsf{X}=(\vect{x}_1,\cdots,\vect{x}_D,
\mathsf{\Sigma}_x\vect{x}_1^\ast,\cdots,\mathsf{\Sigma}_x\vect{x}_D^\ast)$,
Eq.~(\ref{eq:RPAeq-c}) yields Eqs.~(\ref{eq:RPAeq-b}) and (\ref{eq:norm-b})
for $\nu=1,\cdots,D$.\\
This part of the proposition is proven also from Eq.~(\ref{eq:norm-relation})
with $\nu=\nu'$ and Corollary~\ref{cor:non-Jordan}.
If $\mathsf{S}$ is positive-definite,
$\vect{x}_\nu^\dagger\,\mathsf{S}\,\vect{x}_\nu>0$,
deriving $\omega_\nu>0$ and $\vect{x}_\nu^\dagger\,\mathsf{N}\,\vect{x}_\nu>0$
or both negative, for any solution of the RPA equation (\ref{eq:RPAeq-b}).
Then Corollary~\ref{cor:non-Jordan} ensures
that no eigenvector constitutes Jordan blocks.\\
Conversely, if the RPA equation is fully solvable as in Eq.~(\ref{eq:RPAeq-c}),
$\det\mathsf{X}\ne 0$ and $\omega_\nu>0$ for $^\forall\nu\,(=1,\cdots,D)$.
Then, since
\be \mathsf{S}=\mathsf{U}^{-1}\,\mathrm{diag.}(\lambda_i)\,\mathsf{U}
 =(\mathsf{N\,X\,N})\,\mathsf{\Omega}\,(\mathsf{N\,X\,N})^\dagger\,,\ee
it follows that
\be \mathrm{diag.}(\lambda_i)
 =(\mathsf{U\,N\,X\,N})\,\mathsf{\Omega}\,(\mathsf{U\,N\,X\,N})^\dagger
\ee
and therefore, by expressing $\mathsf{U\,N\,X\,N}=(\chi_{i\nu})$,
\be \lambda_i =\sum_\nu \omega_\nu\,|\chi_{i\nu}|^2 >0\,;
\label{eq:pos-def-reverse}\ee
namely $\mathsf{S}$ is positive-definite.
Notice $\big|\det(\mathsf{U\,N\,X\,N})\big|=\big|\det\,\mathsf{X}\big|\ne 0$,
which excludes the possibility of $\lambda_i=0$
in Eq.~(\ref{eq:pos-def-reverse}).
\end{proof}

\noindent
Therefore the arguments on physical solutions by Thouless are applicable
even under the presence of degeneracy as in the continuum.

For solutions of Class~\ref{item:negative},
a positive eigenvalue $\omega_\nu$ is accompanied by
an eigenvector $\vect{x}_\nu$ with $\vect{x}_\nu^\dagger\mathsf{N}\,\vect{x}_\nu<0$
or $\vect{x}_\nu^\dagger\mathsf{N}\,\vect{x}_\nu=0$.
In the former case, its UL-dual partner is normalizable
but corresponds to the eigenvalue $-\omega_\nu\,(<0)$.
The submatrix $\mathsf{S}_{[\nu]}$ of this solution is negative-definite,
as exemplified in Appendix~\ref{app:2*2}.
Therefore the stability matrix $\mathsf{S}$
has two negative eigenvalues at least.
In the $\vect{x}_\nu^\dagger\mathsf{N}\,\vect{x}_\nu=0$ case,
$\vect{x}_\nu$ forms a Jordan block~\cite{ref:Neergard},
whose UL-dual partner associated with $-\omega_\nu\,(<0)$
belongs to another Jordan block.
Probably for this reason, Thouless ignored this class of solutions,
having focused on his arguments near the stability.

\subsection{Solutions for complex eigenvalues}\label{subsec:imag-sol}

Complex eigenvalues belong to Classes~\ref{item:pure-imag}
and \ref{item:quartet} of Prop.~\ref{prop:eigenvalues}.
I next discuss properties of solutions of Class~\ref{item:pure-imag}.

\begin{lemma}\label{lem:pure-imag-eigen}
For an eigenvalue $\omega_\nu$ of $\mathsf{N\,S}$
with $\mathrm{Re}(\omega_\nu)=0$,
any corresponding eigenvector can be taken so as to satisfy
$\mathsf{\Sigma}_x\vect{x}_\nu^\ast=e^{-i\phi}\,\vect{x}_\nu$ ($\phi\in\mathbf{R}$).
Conversely, if an eigenvector $\vect{x}_\nu$ satisfies
$\mathsf{\Sigma}_x\vect{x}_\nu^\ast=e^{-i\phi}\,\vect{x}_\nu$,
$\mathrm{Re}(\omega_\nu)=0$ holds for its corresponding eigenvalue.
\end{lemma}
\begin{proof}
$\mathrm{Re}(\omega_\nu)=0$ is equivalent to $\omega_\nu=-\omega_\nu^\ast$.
Therefore, from Corollary~\ref{lem:sym-sigma},
both $\vect{x}_\nu$ and $\mathsf{\Sigma}_x\vect{x}_\nu^\ast$
belong to the equal eigenvalue $\omega_\nu$,
whether they are linearly dependent or independent.
Then a linear combination of them,
$\vect{y}_\nu:=\alpha\,\vect{x}_\nu + \beta\,\mathsf{\Sigma}_x\vect{x}_\nu^\ast$
($\alpha,\beta\in\mathbf{C}$),
is also an eigenvector associated with the eigenvalue $\omega_\nu$.
Assuming $\beta=e^{i\phi}\,\alpha^\ast$,
we verify $\mathsf{\Sigma}_x\vect{y}_\nu^\ast=e^{-i\phi}\,\vect{y}_\nu$.
When $\vect{x}_\nu$ and $\mathsf{\Sigma}_x\vect{x}_\nu^\ast$
are linearly independent,
we obtain two independent vectors
by adopting , \textit{e.g.}, $\alpha=-\beta=1$ and $\alpha=\beta=i$.\\
If $\omega_\nu\ne-\omega_\nu^\ast$,
the associating eigenvectors
$\vect{x}_\nu$ and $\mathsf{\Sigma}_x\vect{x}_\nu^\ast$
must be linearly independent;
namely $\mathsf{\Sigma}_x\vect{x}_\nu^\ast=e^{-i\phi}\,\vect{x}_\nu$ is impossible.
\end{proof}
\noindent
It is noted here that $\mathrm{Re}(\omega_\nu)=0$ covers
the solutions of Class~\ref{item:null} as well as \ref{item:pure-imag}
in Prop.~\ref{prop:eigenvalues}.

If $\mathsf{\Sigma}_x\vect{x}_\nu^\ast=e^{-i\phi}\,\vect{x}_\nu$ is assumed
for the $\mathrm{Re}(\omega_\nu)=0$ case,
the lower $D$-dimensional components of the RPA equation (\ref{eq:RPAeq-b})
is only a repetition of the upper $D$-dimensional components.
We can choose the phase $e^{i\phi}$ arbitrarily,
because it is controllable via a transformation
$\vect{y}_\nu=e^{i\theta}\vect{x}_\nu$.
A convenient choice is $e^{i\phi}=-1$, so that ${\displaystyle\vect{x}_\nu
= -\mathsf{\Sigma}_x\vect{x}_\nu^\ast
= \begin{pmatrix} X^{(\nu)}\\ -X^{(\nu)\ast}\end{pmatrix}}$.

Let us now take $c_k^{(\nu)}\in\mathbf{R}$ in Eq.~(\ref{eq:Jordan}).
\begin{lemma}\label{lem:pure-imag-Jordan}
For an eigenvalue $\omega_\nu$ of $\mathsf{N\,S}$
with $\mathrm{Re}(\omega_\nu)=0$,
all corresponding basis vectors can satisfy
$\mathsf{\Sigma}_x\vect{\xi}_k^{(\nu)\ast}=-\vect{\xi}_k^{(\nu)}$;
\textit{i.e.}, ${\displaystyle\vect{\xi}_k^{(\nu)}
 = \begin{pmatrix} \Xi^{(\nu,k)}\\ -\Xi^{(\nu,k)\ast}\end{pmatrix}}$.
\end{lemma}
\noindent
This lemma states that solutions of Classes~\ref{item:pure-imag}
and \ref{item:null} can be self UL-dual,
together with basis vectors generated from them.
\begin{proof}
In the case that $\mathrm{Re}(\omega_\nu)=0$,
Eqs.~(\ref{eq:Jordan}) and (\ref{eq:Jordan-sym}) come
\be\begin{split}
 (\mathsf{S}-\omega_\nu\,\mathsf{N})\,\vect{\xi}_{k+1}^{(\nu)}
 &= ic_k^{(\nu)}\,\mathsf{N}\,\vect{\xi}_k^{(\nu)}\,,\\
 (\mathsf{S}-\omega_\nu\,\mathsf{N})\,\mathsf{\Sigma}_x\vect{\xi}_{k+1}^{(\nu)\ast}
 &= ic_k^{(\nu)}\,\mathsf{N}\,\mathsf{\Sigma}_x\vect{\xi}_k^{(\nu)\ast}\,.
\end{split}\label{eq:Jordan-sym2}\ee
From Lemma~\ref{lem:pure-imag-eigen} and the argument above,
we can assume $\mathsf{\Sigma}_x\vect{\xi}_1^{(\nu)\ast}=-\vect{\xi}_1^{(\nu)}$.
If $\mathsf{\Sigma}_x\vect{\xi}_k^{(\nu)\ast}=-\vect{\xi}_k^{(\nu)}$,
the first equation of (\ref{eq:Jordan-sym2}) indicates
that the second equation has a solution fulfilling
$\mathsf{\Sigma}_x\vect{\xi}_{k+1}^{(\nu)\ast}=-\vect{\xi}_{k+1}^{(\nu)}$.
The lemma is then proven inductively.
\end{proof}
\noindent
Compatibility of this lemma with Prop.~\ref{theor:one-to-one_basis}
is confirmed in Appendix~\ref{app:sub-B1}.
Under the above convention for the $\mathrm{Re}(\omega_\nu)=0$ case,
the normalization condition of $\vect{\xi}_{k}^{(\nu)}$ can be
\be \bbar{\vect{\xi}}_k^{(\nu)\dagger}\,\mathsf{N}\,\vect{\xi}_{k'}^{(\nu')}
 = \bar{\Xi}^{(\nu,k)\dagger}\,\Xi^{(\nu',k')}
 - \big[\bar{\Xi}^{(\nu,k)\dagger}\,\Xi^{(\nu',k')}\big]^\ast
 = \pm i\delta_{\nu\nu'}\,\delta_{kk'}\,,
\label{eq:imag-norm}\ee
although $\vect{\xi}_{k}^{(\nu)}$ is not normalizable
in the respect of Def.~\ref{def:normalizable}.

Consider solutions in vicinity of the stability,
in which the stability matrix $\mathsf{S}$ has a single negative eigenvalue.
Near the stability,
the subspace providing negative $\det\,\mathsf{S}_{[\nu]}$ can be separated out
by using the projector in Sec.~\ref{sec:decomp},
which should be two dimensional and therefore provides
${\displaystyle\mathsf{S}_{[\nu]}=\begin{pmatrix} a& b\\
b^\ast& a\end{pmatrix}}$ ($a\in\mathbf{R}$, $b\in\mathbf{C}$),
as in Appendix~\ref{app:2*2}.
A pair of pure-imaginary eigenvalues is obtained,
illustrating that the first unphysical solution
emerges as Class~\ref{item:pure-imag} of Prop.~\ref{prop:eigenvalues}.

Unlike the self LR-duality for real eigenvalues,
the self UL-duality does not forbid Jordan blocks,
although most pure-imaginary eigenvalues are expected
not to form Jordan blocks.
An example of Jordan blocks is presented in Appendix~\ref{app:4*4-pure-imag}.

Let us turn to solutions of Class~\ref{item:quartet}.
Quartet solutions are a manifestation of the two types of dualities.
The possibility of quartet solutions
was first pointed out in Ref.~\cite{ref:UR71}
for $\mathsf{S}=\mathsf{S}^\ast$ cases,
and mentioned in Ref.~\cite{ref:SM84} in more general context.
A minimal model for quartet solutions is constructed by taking $D=2$,
and is analyzed in Appendix~\ref{app:4*4-quartet}.
For quartet solutions $\nu$, $d_{[\nu]}$ is a multiple of four.
Hence, by denoting the solutions $\pm\alpha\pm i\beta$
($\alpha,\beta\in\mathbf{R}$),
$\det\,\mathsf{S}_{[\nu]}=\det(\mathsf{N}_{[\nu]}\,\mathsf{S}_{[\nu]})
=(\alpha^2+\beta^2)^{d_{[\nu]}/2}>0$.
As $\mathsf{S}$ cannot be positive-definite
on account of the latter part of Prop.~\ref{theor:stability},
Eq.~(\ref{eq:detS-decomp}) indicates
that $\mathsf{S}$ has at least two negative eigenvalues
for quartet solutions to come out.

\subsection{NG-mode solutions}\label{subsec:NG-sol}

The simplest example of the NG-mode solution
is given in Appendix~\ref{app:2*2}, by the $2\times 2$ stability matrix.
It illustrates that the null eigenvalue
is often associated with a two-dimensional Jordan block,
as indicated by Thouless~\cite{ref:Thou61}.
The NG modes that generate two-dimensional Jordan blocks
have well been investigated~\cite{ref:RS80,ref:Thou61,ref:TV62}.
However, in the example of Appendix~\ref{app:2*2},
there is a trivial case of $\mathsf{S}=\mathsf{0}$
in which two $d_\nu=1$ eigenvectors are present for the null eigenvalues.
Moreover, an example of 4-dimensional Jordan block
is seen in Appendix~\ref{app:4*4-null}.
Likely or not, it is difficult to exclude the possibilities
other than the two-dimensional Jordan block for the null eigenvalue
only from mathematical viewpoints.

\begin{corollary}\label{cor:deg-null}
If there exists a null eigenvalue for $\mathsf{N\,S}$,
it must have even number of degeneracy.
\end{corollary}
\begin{proof}
Because of Prop.~\ref{theor:sym-eigen},
the number of non-zero eigenvalues must be even, up to their degeneracies.
Moreover, Lemma~\ref{lem:Jordan-sym} ensures
that sum of dimensions of Jordan blocks for non-zero eigenvalues is even.
The total dimension of $\mathsf{N\,S}$ is $2D$,
which concludes the degeneracy of the null eigenvalue must be even.
Also proven from Lemma~\ref{lem:dim-Lambda}.
\end{proof}

When SSB occurs, there must be NG-mode solutions
corresponding to the broken symmetry;
\textit{e.g.}, the linear momentum in the SSB with respect to the translation
and the angular momentum in the SSB with respect to the rotation
in deformed nuclei.
For specific NG-mode solutions with such physical interpretations,
their properties can further be explored,
though I do not pursue this direction in this article.

The null eigenvalues may lie at the intersection
of the self LR- and the self UL-dualities.
Although there is no single eigenvector having both of the self dualities
as indicated by Lemma~\ref{lem:dim-Lambda},
there could be an even-dimensional Jordan block
in which the LR-duality closes by its basis vectors,
while keeping the self UL-duality of Lemma~\ref{lem:pure-imag-Jordan}.
In such cases the Jordan block, instead of the basis vectors,
may be said self LR-dual.
\begin{proposition}\label{theor:null-self-dual}
For even-dimensional Jordan blocks associated with a null eigenvalue,
it is possible to produce basis vectors $\big\{\vect{\eta}_k;
k=1,\cdots,d_\nu\big\}$ having double self duality,
$\mathsf{\Sigma}_x\vect{\eta}_k^\ast=-\vect{\eta}_k$ and
$\bbar{\vect{\eta}}_k=\vect{\eta}_{d_\nu+1-k}$.
\end{proposition}
\noindent
This proposition is proven in Appendix~\ref{app:sub-B2}.
An example of the transformation shown in Appendix~\ref{app:sub-B2}
is given by the NG mode of the angular momentum,
under SSB with respect to the rotation.
Even though the Jordan blocks corresponding to $J_\pm$
are the LR-dual of each other,
their linear combinations provide those corresponding to $J_x$ and $J_y$,
each of which could be self LR-dual.

Let us focus on the case that the basis vectors simultaneously fulfill
$\mathsf{\Sigma}_x\vect{\xi}_k^{(\nu)\ast}=-\vect{\xi}_k^{(\nu)}$
and $\bbar{\vect{\xi}}_k^{(\nu)}=\vect{\xi}_{d_\nu+1-k}^{(\nu)}$, with even $d_\nu$.
Then the projector of Eq.~(\ref{eq:proj-nuk})
has the relation $\mathsf{N\,\Lambda}_{\nu,k}^\dagger\,\mathsf{N}
=\mathsf{\Lambda}_{\nu,d_\nu+1-k}$,
and therefore obeys
\be\begin{split} &\mathsf{\Lambda}_{\nu,k}+\mathsf{\Lambda}_{\nu,d_\nu+1-k}
 =\mathsf{N}\,\big[\mathsf{\Lambda}_{\nu,k}
 +\mathsf{\Lambda}_{\nu,d_\nu+1-k}\big]^\dagger\,\mathsf{N}\,,\\
&\mathsf{\Sigma}_x\,\big[\mathsf{\Lambda}_{\nu,k}
 +\mathsf{\Lambda}_{\nu,d_\nu+1-k}\big]^\ast\,\mathsf{\Sigma}_x
 = \mathsf{\Lambda}_{\nu,k}+\mathsf{\Lambda}_{\nu,d_\nu+1-k}\,,
\label{eq:null-proj2}\end{split}\ee
although the relation analogous to Prop.~\ref{theor:SN-proj}
(Eq.~(\ref{eq:Snu-separation}), in particular)
does not necessarily hold.
Like the arguments using $\mathsf{\Lambda}_{[\nu]}$ in Sec.~\ref{sec:decomp},
$\mathsf{\Lambda}_{\nu,k}+\mathsf{\Lambda}_{\nu,d_\nu+1-k}$
produces a subspace keeping both the LR- and UL-dualities.
For doubly self-dual Jordan blocks,
the basis vectors for the NG-mode may be removed two by two via the projectors,
with minimal d.o.f. coupled to them.
Within the two-dimensional subspace
given by $\mathsf{\Lambda}_{\nu,d_\nu/2}+\mathsf{\Lambda}_{\nu,d_\nu/2+1}$,
the prescription proposed in Refs.~\cite{ref:Thou61,ref:TV62}
(and well summarized in Refs.~\cite{ref:RS80,ref:Row70})
is applicable.

If $d_\nu=\mathrm{odd}$, there should be two Jordan blocks,
which are the partner of the LR- (or UL-) duality of each other.
To separate them, one may apply the projector $\Lambda_{[\nu]}$
introduced in Sec.~\ref{sec:decomp}.

\section{Summary}\label{sec:summary}

Properties of solutions of the RPA equation is reanalyzed
in terms of the linear algebra.
As well as eigensolutions, cases in which the matrix $\mathsf{N\,S}$
(and $\mathsf{S\,N}$) forms Jordan blocks are examined.
Two types of dualities of eigenvectors and basis vectors,
which are called LR- and UL-dualities in this article,
are pointed out and explored.
These dualities are useful to clarify properties of the RPA solutions.
Projection respecting the dualities is developed.

Eigenvalues given by the RPA equation are classified
into five classes, in Prop.~\ref{prop:eigenvalues}.
As pointed out by Thouless, all solutions are physical ones
if the stability matrix is positive-definite.
Its opposite is also true (in absence of NG modes),
being useful to judge stability of a MF solution
from numerical calculations in the RPA.
These solutions are singled out, not constituting Jordan blocks,
and have the self LR-duality while are paired by the UL-duality.
Eigenvectors and basis vectors for pure-imaginary eigenvalues
can be made self UL-dual, and paired by the LR-duality.
With no self dualities, quartet solutions manifest two types of the dualities.
NG-mode solutions, which are associated with the null eigenvalue
and often related to the spontaneous symmetry breaking,
lie at intersection of the two self dualities.
However, a single vector cannot be both self LR-dual and self UL-dual.
Only even-dimensional Jordan blocks can have double self dualities.
The well-known prescription of separating out the NG modes
could be applicable to such cases.

\section*{Acknowledgment}

The author is grateful to K.~Matsuyanagi, K.~Neerg\aa rd, H.~Kurasawa, Y.R.~Shimizu,
J.~Terasaki and T.~Inakura for discussions.
This work is financially supported in part
by JSPS KAKENHI Grant Number~24105008 and Grant Number~16K05342.

\appendix
\section{Explicit proofs of proposition and lemmas
in Sec.~\protect\ref{subsec:Jordan}}
\label{app:proof-Jordan}

Analogously to Eq.~(\ref{eq:norm-relation}),
the following relation is obtained from Eq.~(\ref{eq:Jordan}),
\be\begin{split}
 \vect{\xi}_{k+1}^{(\nu)\dagger}\,\mathsf{S}\,\vect{\xi}_{k'+1}^{(\nu')}
 &= \omega_\nu^\ast\,\vect{\xi}_{k+1}^{(\nu)\dagger}\,\mathsf{N}\,\vect{\xi}_{k'+1}^{(\nu')}
  - ic_k^{(\nu)\ast}\,\vect{\xi}_k^{(\nu)\dagger}\,\mathsf{N}\,
    \vect{\xi}_{k'+1}^{(\nu')}\\
 &= \omega_{\nu'}\,\vect{\xi}_{k+1}^{(\nu)\dagger}\,\mathsf{N}\,\vect{\xi}_{k'+1}^{(\nu')}
  + ic_{k'}^{(\nu')}\,\vect{\xi}_{k+1}^{(\nu)\dagger}\,\mathsf{N}\,
    \vect{\xi}_{k'}^{(\nu')}\,.
\end{split}\label{eq:Jordan-orthogonality}\ee
While Eq.~(\ref{eq:Jordan}) runs for $k=1,\cdots,d_\nu-1$,
it can be extended to $k=0$ and $d_\nu$
if we assume $\vect{\xi}_0^{(\nu)}=\vect{\xi}_{d_\nu+1}^{(\nu)}=\vect{0}$
and $c_{d_\nu}^{(\nu)}=0$.
Equation~(\ref{eq:Jordan-orthogonality}) is also extended accordingly.

\subsection{Proof of Lemma~\protect\ref{lem:Jordan-orthogonality}}
\label{app:sub-A1}

\begin{proof}
Suppose $\omega_\nu^\ast\ne\omega_{\nu'}$.
Let us first take $k=k'=0$ in Eq.~(\ref{eq:Jordan-orthogonality}).
As Lemma~\ref{lem:bi-orthogonality} states,
$\vect{\xi}_1^{(\nu)\dagger}\,\mathsf{N}\,\vect{\xi}_1^{(\nu')}
=\vect{x}_\nu^\dagger\,\mathsf{N}\,\vect{x}_{\nu'}=0$ follows.
We then increase $k'$, with keeping $k=0$.
For $k'=1$, Eq.~(\ref{eq:Jordan-orthogonality}) yields
\be\begin{split}
 \omega_\nu^\ast\,\vect{\xi}_1^{(\nu)\dagger}\,\mathsf{N}\,\vect{\xi}_2^{(\nu')}
 &= \omega_{\nu'}\,\vect{\xi}_1^{(\nu)\dagger}\,\mathsf{N}\,\vect{\xi}_2^{(\nu')}
  + ic_1^{(\nu')}\,\vect{\xi}_1^{(\nu)\dagger}\,\mathsf{N}\,\vect{\xi}_1^{(\nu')}\\
 &= \omega_{\nu'}\,\vect{\xi}_1^{(\nu)\dagger}\,\mathsf{N}\,\vect{\xi}_2^{(\nu')}\,.
\end{split}\label{eq:Jordan-orthogonality1}\ee
This derives $\vect{\xi}_1^{(\nu)\dagger}\,\mathsf{N}\,\vect{\xi}_2^{(\nu')}=0$.
Thus, by successively applying Eq.~(\ref{eq:Jordan-orthogonality}),
$\vect{\xi}_1^{(\nu)\dagger}\,\mathsf{N}\,\vect{\xi}_{k'}^{(\nu')}=0$ follows
for all $k'$.\\
Let us next increase $k$.
For $k=1$, Eq.~(\ref{eq:Jordan-orthogonality}) leads to
\be\begin{split}
 \omega_\nu^\ast\,\vect{\xi}_2^{(\nu)\dagger}\,\mathsf{N}\,\vect{\xi}_{k'+1}^{(\nu')}
  - ic_1^{(\nu)\ast}\,\vect{\xi}_1^{(\nu)\dagger}\,\mathsf{N}\,
    \vect{\xi}_{k'+1}^{(\nu')}
 &= \omega_{\nu'}\,\vect{\xi}_2^{(\nu)\dagger}\,\mathsf{N}\,\vect{\xi}_{k'+1}^{(\nu')}
  + ic_{k'}^{(\nu')}\,\vect{\xi}_2^{(\nu)\dagger}\,\mathsf{N}\,\vect{\xi}_{k'}^{(\nu')}\\
 =\omega_\nu^\ast\,\vect{\xi}_2^{(\nu)\dagger}\,\mathsf{N}\,\vect{\xi}_{k'+1}^{(\nu')}
 &\,.
\end{split}\label{eq:Jordan-orthogonality2}\ee
In a similar manner to the $k=0$ case,
$\vect{\xi}_2^{(\nu)\dagger}\,\mathsf{N}\,\vect{\xi}_{k'}^{(\nu')}=0$ can be shown
for all $k'$
by applying Eq.~(\ref{eq:Jordan-orthogonality2})
with increasing $k'$ from $k'=0$.
The lemma is proven by repeating this process for increasing $k$.
\end{proof}

\subsection{Proof of Lemma~\protect\ref{lem:Jordan-normalizability}}
\label{app:sub-A2}

\begin{proof}
Suppose $\omega_\nu^\ast=\omega_{\nu'}$.
In this case Eq.~(\ref{eq:Jordan-orthogonality}) yields
\be -c_k^{(\nu)\ast}\,\vect{\xi}_k^{(\nu)\dagger}\,\mathsf{N}\,
    \vect{\xi}_{k'+1}^{(\nu')}
 = c_{k'}^{(\nu')}\,\vect{\xi}_{k+1}^{(\nu)\dagger}\,\mathsf{N}\,
    \vect{\xi}_{k'}^{(\nu')}\,.
\label{eq:Jordan-orthogonality3}\ee
Corresponding to the $k=0$ case in Eq.~(\ref{eq:Jordan-orthogonality3}),
$c_{k'}^{(\nu')}\,\vect{\xi}_1^{(\nu)\dagger}\,\mathsf{N}\,\vect{\xi}_{k'}^{(\nu')}=0$.
This concludes that $\vect{\xi}_1^{(\nu)\dagger}\,\mathsf{N}\,
\vect{\xi}_{k'}^{(\nu')}$ can be non-zero
only when $\vect{\xi}_{k'+1}^{(\nu')}$ does not exist.
Namely, $\vect{x}_\nu=\vect{\xi}_1^{(\nu)}$ can overlap
only with the last basis vector $\mathsf{N}\,\vect{\xi}_{d_{\nu'}}^{(\nu')}$
of the Jordan block of $\mathsf{S\,N}$.\\
By setting $k=1$,
$-c_1^{(\nu)\ast}\,\vect{\xi}_1^{(\nu)\dagger}\,\mathsf{N}\,
    \vect{\xi}_{k'+1}^{(\nu')}
 = c_{k'}^{(\nu')}\,\vect{\xi}_2^{(\nu)\dagger}\,\mathsf{N}\,
    \vect{\xi}_{k'}^{(\nu')}$ is obtained.
The l.h.s. vanishes unless $c_1^{(\nu)}\ne 0$ and $k'=d_{\nu'}-1$.
Therefore, for $d_\nu\geq 2$, $\vect{\xi}_2^{(\nu)}$ can overlap
only with the two last basis vectors $\mathsf{N}\,\vect{\xi}_{d_{\nu'}-1}^{(\nu')}$
and $\mathsf{N}\,\vect{\xi}_{d_{\nu'}}^{(\nu')}$
of the Jordan block of $\mathsf{S\,N}$.
If it is non-zero,
$\vect{\xi}_2^{(\nu)\dagger}\,\mathsf{N}\,\vect{\xi}_{d_{\nu'}}^{(\nu')}$
can be made zero
by a proper transformation $\vect{\xi}_2^{(\nu)}\,\rightarrow\,
\vect{\xi}_2^{(\nu)}+\beta\,\vect{\xi}_1^{(\nu)}$\,($\beta\in\mathbf{C}$).
$\vect{\xi}_k^{(\nu)}$ with $k\geq 3$ is also transformed, accordingly.
The first part of the lemma is proven by repeating this argument.\\
The additional part of the lemma concerns the case
in which plural Jordan blocks have an equal eigenvalue $\omega_\nu$.
If $d_\nu>d_{\nu'}$, it is concluded
by continuing the above argument until reaching $k'=1$
that $\vect{\xi}_{d_{\nu'}}^{(\nu)\dagger}\,\mathsf{N}\,\vect{\xi}_1^{(\nu')}$
can be non-zero,
but $\vect{\xi}_k^{(\nu)\dagger}\,\mathsf{N}\,\vect{\xi}_1^{(\nu')}=0$
for $k>d_{\nu'}$ including $k=d_\nu$.
By reversing the above argument,
this indicates $\vect{\xi}_k^{(\nu)\dagger}\,\mathsf{N}\,\vect{\xi}_1^{(\nu')}=0$
for any $k$,
and then $\vect{\xi}_k^{(\nu)\dagger}\,\mathsf{N}\,\vect{\xi}_{k'}^{(\nu')}=0$
for any $k$ and $k'$.
The same holds for the $d_\nu<d_{\nu'}$ case.
\end{proof}

\subsection{Proof of Prop.~\protect\ref{theor:one-to-one_basis}}
\label{app:sub-A3}

\begin{proof}
Since the basis vectors of $\mathsf{S\,N}$ spans a complete set,
any non-vanishing vector overlaps with at least one of them.
Consider the case that some basis vectors of $\mathsf{N\,S}$ have overlap
with plural basis vectors of $\mathsf{S\,N}$.
It is sufficient to consider
that one of the basis vectors of a single Jordan block
(or an eigenvector as a special case)
$\big\{\vect{\xi}_k^{(\nu)}; k=1,\cdots,d_\nu\big\}$ have non-vanishing overlaps
with members of two Jordan blocks (or two eigenvectors) of $\mathsf{S\,N}$,
$\big\{\mathsf{N}\,\vect{\xi}_{k'}^{(\nu')}; k'=1,\cdots,d_{\nu'}\big\}$ and
$\big\{\mathsf{N}\,\vect{\xi}_{k'}^{(\nu'')}; k'=1,\cdots,d_{\nu''}\big\}$.
From Lemmas~\ref{lem:Jordan-orthogonality}
and \ref{lem:Jordan-normalizability},
$\omega_\nu^\ast=\omega_{\nu'}=\omega_{\nu''}$ and $d_\nu=d_{\nu'}=d_{\nu''}$.
Consider a linear combination
$\alpha_{k'}\vect{\xi}_{k'}^{(\nu')}+\beta_{k'}\vect{\xi}_{k'}^{(\nu'')}$
($\alpha_{k'},\beta_{k'}\in\mathbf{C}$).
Equation~(\ref{eq:Jordan}) yields
\be\begin{split} \mathsf{S}\,\big[\alpha_{k'+1}\vect{\xi}_{k'+1}^{(\nu')}
 +\beta_{k'+1}\vect{\xi}_{k'+1}^{(\nu'')}\big]
 &= \omega_\nu\,\mathsf{N}\,\big[\alpha_{k'+1}\vect{\xi}_{k'+1}^{(\nu')}
 +\beta_{k'+1}\vect{\xi}_{k'+1}^{(\nu'')}\big]\\
 &~+ i\,\mathsf{N}\,\big[c_{k'}^{(\nu')}\alpha_{k'+1}\vect{\xi}_{k'}^{(\nu')}
 +c_{k'}^{(\nu'')}\beta_{k'+1}\vect{\xi}_{k'}^{(\nu'')}\big]\,.
\end{split}\ee
By imposing
\be \frac{\alpha_{k'+1}}{\beta_{k'+1}}
 = \frac{\alpha_{k'}/c_{k'}^{(\nu')}}{\beta_{k'}/c_{k'}^{(\nu'')}}\,,
\label{eq:a&b-recursive}\ee
and determining $\alpha_{k'}$ and $\beta_{k'}$ recursively,
$\Big\{\mathsf{N}\,\big[\alpha_{k'}\vect{\xi}_{k'}^{(\nu')}
 +\beta_{k'}\vect{\xi}_{k'}^{(\nu'')}\big]; k'=1,\cdots,d_{\nu'}\Big\}$
also forms a Jordan block of $\mathsf{S\,N}$.
However, we may choose $\alpha_1$ and $\beta_1$
such that $\vect{\xi}_{d_\nu}^{(\nu)\dagger}\,\mathsf{N}\,
\big[\alpha_1\vect{x}_{\nu'}+\beta_1\vect{x}_{\nu''}\big]=0$.
Then, with $\alpha_{k'}$ and $\beta_{k'}$
determined by Eq.~(\ref{eq:a&b-recursive}),
all basis vectors belonging to the Jordan block of $\vect{x}_\nu$
are orthogonal to all vectors in
$\Big\{\mathsf{N}\,\big[\alpha_{k'}\vect{\xi}_{k'}^{(\nu')}
 +\beta_{k'}\vect{\xi}_{k'}^{(\nu'')}\big]; k'=1,\cdots,d_{\nu'}\Big\}$,
as is clear from the argument in Appendix~\ref{app:sub-A1}.
A Jordan block independent of it is left non-orthogonal
to $\big\{\vect{\xi}_k^{(\nu)}; k=1,\cdots,d_\nu\big\}$.
\end{proof}

\subsection{Derivation of Eq.~(\protect\ref{eq:Jordan-inv})}
\label{app:sub-A4}

Expansion of $\mathsf{S}\,\bbar{\vect{\xi}}_{k}^{(\nu)}$
by $\big\{\mathsf{N}\,\bbar{\vect{\xi}}_{k'}^{(\nu')}\big\}$ leads to
\be \mathsf{S}\,\bbar{\vect{\xi}}_k^{(\nu)}
 = \sum_{\nu',k'} \mathsf{\Lambda}_{\nu',k'}^\dagger\,\mathsf{S}\,
 \bbar{\vect{\xi}}_k^{(\nu)}
 = \sum_{\nu',k'}
 \frac{\mathsf{N}\,\bbar{\vect{\xi}}_{k'}^{(\nu')}\,\vect{\xi}_{k'}^{(\nu')\dagger}}
 {\vect{\xi}_{k'}^{(\nu')\dagger}\,\mathsf{N}\,\bbar{\vect{\xi}}_{k'}^{(\nu')}}\,
 \mathsf{S}\,\bbar{\vect{\xi}}_k^{(\nu)}\,,
\label{eq:Jordan-inv-proj}\ee
where $\mathsf{\Lambda}_{\nu,k}$ is the projector
defined in Sec.~\ref{sec:decomp}.
By inserting Eq.~(\ref{eq:Jordan-me}) into Eq.~(\ref{eq:Jordan-inv-proj}),
Eq.~(\ref{eq:Jordan-inv}) is derived.

\section{Additional proofs for Sec.~\protect\ref{sec:stable&unstable}}
\label{app:proof-NG}

\subsection{Verification of compatibility
of Lemma~\protect\ref{lem:pure-imag-Jordan}
with Prop.~\protect\ref{theor:one-to-one_basis}}
\label{app:sub-B1}

Suppose that two sets of basis vectors
$\big\{\mathsf{N}\,\vect{\xi}_k^{(\nu')}; k=1,\cdots,d_{\nu'}\big\}$ and
$\big\{\mathsf{N}\,\vect{\xi}_k^{(\nu'')}; k=1,\cdots,d_{\nu'}\big\}$
associated with $\omega_{\nu'}$, where $\mathrm{Re}(\omega_{\nu'})=0$,
obey the convention of Lemma~\ref{lem:pure-imag-Jordan}.
Since a linear combination
$\alpha_k\vect{\xi}_k^{(\nu')}+\beta_k\vect{\xi}_k^{(\nu'')}$
($\alpha_k,\beta_k\in\mathbf{C}$) yields
\be \mathsf{\Sigma}_x\,
 \big[\alpha_k\vect{\xi}_k^{(\nu')}+\beta_k\vect{\xi}_k^{(\nu'')}\big]^\ast
 = -\big[\alpha_k^\ast\vect{\xi}_k^{(\nu')}+\beta_k^\ast\vect{\xi}_k^{(\nu'')}\big]\,,
\ee
$\alpha_k$ and $\beta_k$ must be real for all $k\,(=1,\cdots,d_\nu)$
in order to ensure that
the proof of Prop.~\ref{theor:one-to-one_basis} in Appendix~\ref{app:sub-A3}
is applicable without influencing the convention
of Lemma~\ref{lem:pure-imag-Jordan}.
The condition for this to be possible is
that $(\vect{\xi}_{d_\nu}^{(\nu)\dagger}\,\mathsf{N}\,\vect{x}_{\nu'})\big/
(\vect{\xi}_{d_\nu}^{(\nu)\dagger}\,\mathsf{N}\,\vect{x}_{\nu''})$ is real
for a basis vector $\vect{\xi}_k^{(\nu)}$ in \ref{app:sub-A3},
and that $c_k^{(\nu')}/c_k^{(\nu'')}$ is real, at the same time.
The former is the condition for
$\vect{\xi}_{d_\nu}^{(\nu)\dagger}\,\mathsf{N}\,
\big[\alpha_1\vect{x}_{\nu'}+\beta_1\vect{x}_{\nu''}\big]=0$ to be possible
with real $\alpha_1$ and $\beta_1$.
This is satisfied in the convention of Lemma~\ref{lem:pure-imag-Jordan},
because $\vect{\xi}_{d_\nu}^{(\nu)\dagger}\,\mathsf{N}\,\vect{x}_{\nu'}$
and $\vect{\xi}_{d_\nu}^{(\nu)\dagger}\,\mathsf{N}\,\vect{x}_{\nu''}$
are pure imaginary.
The latter is the condition
for the recursive relation of Eq.~(\ref{eq:a&b-recursive})
not to break $\alpha_k,\beta_k\in\mathbf{R}$,
and is fulfilled by the convention $c_k^{(\nu')},c_k^{(\nu'')}\in\mathbf{R}$.

\subsection{Proof of Prop.~\protect\ref{theor:null-self-dual}}
\label{app:sub-B2}

\begin{proof}
In the cases under discussion,
both $\vect{x}_\nu$ and $\bbar{\vect{x}}_\nu\,(=\bbar{\vect{\xi}}_{d_\nu}^{(\nu)})$
belong to the same eigenvalue that is zero.
Equations~(\ref{eq:Jordan}) and (\ref{eq:Jordan-inv}) become
\be\begin{split}
 \mathsf{S}\,\vect{\xi}_{k+1}^{(\nu)}
 &= ic_k^{(\nu)}\,\mathsf{N}\,\vect{\xi}_k^{(\nu)}\,,\\
 \mathsf{S}\,\bbar{\vect{\xi}}_{d_\nu-k}^{(\nu)}
 &= - ic_{d_\nu-k}^{(\nu)\ast}\,
\frac{\vect{\xi}_{d_\nu-k}^{(\nu)\dagger}\,\mathsf{N}\,\bbar{\vect{\xi}}_{d_\nu-k}^{(\nu)}}
 {\vect{\xi}_{d_\nu+1-k}^{(\nu)\dagger}\,\mathsf{N}\,\bbar{\vect{\xi}}_{d_\nu+1-k}^{(\nu)}}\,
 \mathsf{N}\,\bbar{\vect{\xi}}_{d_\nu+1-k}^{(\nu)}\,.
\end{split}\ee
If $\vect{x}_\nu\ne\bbar{\vect{x}}_\nu$, their linear combination
$\alpha_k\vect{\xi}_k^{(\nu)}+\beta_k\bbar{\vect{\xi}}_{d_\nu+1-k}^{(\nu)}$
($\alpha_k,\beta_k\in\mathbf{C}$) obeys
\be \mathsf{S}\,\big[\alpha_{k+1}\vect{\xi}_{k+1}^{(\nu)}
 +\beta_{k+1}\bbar{\vect{\xi}}_{d_\nu-k}^{(\nu)}\big]
 = i\,\mathsf{N}\,\bigg[c_k^{(\nu)}\,\alpha_{k+1}\vect{\xi}_k^{(\nu)}
 -\frac{\vect{\xi}_{d_\nu-k}^{(\nu)\dagger}\,\mathsf{N}\,
 \bbar{\vect{\xi}}_{d_\nu-k}^{(\nu)}}
  {\vect{\xi}_{d_\nu+1-k}^{(\nu)\dagger}\,\mathsf{N}\,\bbar{\vect{\xi}}_{d_\nu+1-k}^{(\nu)}}
 \,c_{d_\nu-k}^{(\nu)\ast}\beta_{k+1}\bbar{\vect{\xi}}_{d_\nu+1-k}^{(\nu)}\bigg]\,.
\ee
Therefore, for an arbitrary set $(\alpha_1,\beta_1)$,
a Jordan block is obtained by determining $(\alpha_k,\beta_k)$ for $k\geq 2$
from
\be \frac{\alpha_{k+1}}{\beta_{k+1}}
 = -\frac{c_{d_\nu-k}^{(\nu)\ast}}{c_k^{(\nu)}}\,
 \frac{\vect{\xi}_{d_\nu-k}^{(\nu)\dagger}\,\mathsf{N}\,\bbar{\vect{\xi}}_{d_\nu-k}^{(\nu)}}
 {\vect{\xi}_{d_\nu+1-k}^{(\nu)\dagger}\,\mathsf{N}\,\bbar{\vect{\xi}}_{d_\nu+1-k}^{(\nu)}}\,
 \frac{\alpha_k}{\beta_k}\,.
\label{eq:null-dual-recursion}\ee
This reaches
\be \frac{\alpha_{d_\nu}}{\beta_{d_\nu}}
 = (-)^{d_\nu-1}\bigg[\prod_{k=1}^{d_\nu-1}\frac{c_k^{(\nu)\ast}}{c_k^{(\nu)}}\bigg]\,
 \frac{\vect{\xi}_1^{(\nu)\dagger}\,\mathsf{N}\,\bbar{\vect{\xi}}_1^{(\nu)}}
 {\vect{\xi}_{d_\nu}^{(\nu)\dagger}\,\mathsf{N}\,\bbar{\vect{\xi}}_{d_\nu}^{(\nu)}}\,
 \frac{\alpha_1}{\beta_1}\,.
\label{eq:null-dual-coef}\ee
It is assumed here that $\vect{\xi}_k^{(\nu)}$ is self UL-dual
as in Lemma~\ref{lem:pure-imag-Jordan},
and the convention $c_k^{(\nu)}>0$ and
$\bbar{\vect{\xi}}_k^{(\nu)\dagger}\,\mathsf{N}\,\vect{\xi}_k^{(\nu)}=(-)^k\,i$
($k=1,\cdots,d_\nu$) is adopted with even $d_\nu$.
Equation~(\ref{eq:null-dual-recursion}) then comes $\alpha_{k+1}/\beta_{k+1}
=(c_{d_\nu-k}^{(\nu)}/c_k^{(\nu)})\,(\alpha_k/\beta_k)$,
leading to
\be \frac{\alpha_k}{\beta_k}
 = \bigg[\prod_{k'=1}^{k-1}\frac{c_{d_\nu-k'}^{(\nu)}}{c_{k'}^{(\nu)}}\bigg]\,
  \frac{\alpha_1}{\beta_1}
 = \bigg[\frac{\prod_{k'=1}^{d_\nu-1} c_{k'}}
   {\prod_{k'=1}^{k-1} c_{k'}\prod_{k'=1}^{d_\nu-k} c_{k'}}\bigg]\,
  \frac{\alpha_1}{\beta_1}
 = \frac{\alpha_{d_\nu+1-k}}{\beta_{d_\nu+1-k}}\,.
\ee
In particular, Eq.~(\ref{eq:null-dual-coef}) is now
$\alpha_{d_\nu}/\beta_{d_\nu}=\alpha_1/\beta_1$.
Defining new basis vectors by
\be\begin{split}
 \vect{\eta}_k := \frac{1}{\sqrt{2\gamma_k}}\big[\gamma_k\,\vect{\xi}_k^{(\nu)}
    +\bbar{\vect{\xi}}_{d_\nu+1-k}^{(\nu)}\big]\,,\quad&
 \vect{\eta}'_k := \frac{(-)^k}{\sqrt{2\gamma_k}}
 \big[\gamma_k\,\vect{\xi}_k^{(\nu)}-\bbar{\vect{\xi}}_{d_\nu+1-k}^{(\nu)}\big]\,;\\
 &\quad \gamma_1:=1\,,\quad
 \gamma_k:=\prod_{k'=1}^{k-1}\frac{c_{d_\nu-k'}^{(\nu)}}{c_{k'}^{(\nu)}}\quad
 (\mbox{for $k\geq 2$})\,,
\end{split}\ee
we have
\be\begin{split}
 \vect{\eta}_{d_\nu+1-k}^\dagger\,\mathsf{N}\,\vect{\eta}_k
 &= \frac{1}{2\gamma_k}\,\big[\gamma_k\,\vect{\xi}_{d_\nu+1-k}^{(\nu)\dagger}
   +\bbar{\vect{\xi}}_k^{(\nu)\dagger}\big]\,\mathsf{N}\,
 \big[\gamma_k\,\vect{\xi}_k^{(\nu)}+\bbar{\vect{\xi}}_{d_\nu+1-k}^{(\nu)}\big]\\
 &= \frac{1}{2}\,
 \big[\,\bbar{\vect{\xi}}_k^{(\nu)\dagger}\,\mathsf{N}\,\vect{\xi}_k^{(\nu)}
 + \vect{\xi}_{d_\nu+1-k}^{(\nu)\dagger}\,\mathsf{N}\,
 \bbar{\vect{\xi}}_{d_\nu+1-k}^{(\nu)}\big]
 = (-)^k\,i\,,\\
 \vect{\eta}_{d_\nu+1-k}^{\prime\dagger}\,\mathsf{N}\,\vect{\eta}'_k
 &= -\frac{1}{2\gamma_k}\,
 \big[\gamma_k\,\vect{\xi}_{d_\nu+1-k}^{(\nu)\dagger}
   -\bbar{\vect{\xi}}_k^{(\nu)\dagger}\big]\,\mathsf{N}\,
 \big[\gamma_k\,\vect{\xi}_k^{(\nu)}-\bbar{\vect{\xi}}_{d_\nu+1-k}^{(\nu)}\big]\\
 &= \frac{1}{2}\,
 \big[\,\bbar{\vect{\xi}}_k^{(\nu)\dagger}\,\mathsf{N}\,\vect{\xi}_k^{(\nu)}
 + \vect{\xi}_{d_\nu+1-k}^{(\nu)\dagger}\,\mathsf{N}\,
 \bbar{\vect{\xi}}_{d_\nu+1-k}^{(\nu)}\big]
 = (-)^k\,i\,,\\
 \vect{\eta}_{d_\nu+1-k}^\dagger\,\mathsf{N}\,\vect{\eta}'_k
 &= \frac{(-)^k}{2\gamma_k}\,
 \big[\gamma_k\,\vect{\xi}_{d_\nu+1-k}^{(\nu)\dagger}
   +\bbar{\vect{\xi}}_k^{(\nu)\dagger}\big]\,\mathsf{N}\,
 \big[\gamma_k\,\vect{\xi}_k^{(\nu)}-\bbar{\vect{\xi}}_{d_\nu+1-k}^{(\nu)}\big]\\
 &= \frac{(-)^k}{2}\,
 \big[\,\bbar{\vect{\xi}}_k^{(\nu)\dagger}\,\mathsf{N}\,\vect{\xi}_k^{(\nu)}
 - \vect{\xi}_{d_\nu+1-k}^{(\nu)\dagger}\,\mathsf{N}\,
 \bbar{\vect{\xi}}_{d_\nu+1-k}^{(\nu)}\big]
 = 0\,.
\end{split}\label{eq:null-LR}\ee
Note that $\gamma_k$ satisfies $\gamma_{d_\nu+1-k}=\gamma_k$.
Equation~(\ref{eq:null-LR}) allows us to identify
$\bbar{\vect{\eta}}_k=\vect{\eta}_{d_\nu+1-k}$
and $\bbar{\vect{\eta}}'_k=\vect{\eta}'_{d_\nu+1-k}$,
indicating that the Jordan blocks of $\vect{\eta}_k$ and $\vect{\eta}'_k$
have self LR-duality,
with keeping the UL-duality for each $\vect{\eta}_k$ and $\vect{\eta}'_k$.
\end{proof}

\section{Simple examples}

\subsection{$2\times 2$ stability matrix}
\label{app:2*2}

It is instructive to consider an example of $D=1$,
in which all classes of the solutions of Prop.~\ref{prop:eigenvalues}
except \ref{item:quartet} come out.
In this case, the stability matrix $\mathsf{S}$ is set to be
\be \mathsf{S} = \begin{pmatrix} a & b \\
 b^\ast & a \end{pmatrix}\,;\quad
 (a\in\mathbf{R},~b\in\mathbf{C}). \ee
It is easy to recognize that the solutions are categorized according to
Prop.~\ref{prop:eigenvalues} in the following manner,
in correspondence to eigenvalues of $\mathsf{S}$:
\begin{center}\begin{tabular}{cll}
 \protect\ref{item:positive}  & for $a>|b|\geq 0$
 &($\mathsf{S}$ is positive-definite)\,,\\
 \protect\ref{item:negative} & for $-a>|b|\geq 0$
 &($\mathsf{S}$ is negative-definite)\,,\\
 \protect\ref{item:pure-imag} & for $|a|<|b|$
 &($\mathsf{S}$ has a positive and a negative eigenvalues) \,,\\
 \protect\ref{item:null} & for $|a|=|b|$
 &($\mathsf{S}$ has one or two null eigenvalues)\,.\\
\end{tabular}\end{center}
Lemma~\ref{lem:pure-imag-eigen} is confirmed
for the eigenvectors in the case of \ref{item:pure-imag}.
For \ref{item:null}, $\mathsf{N\,S}$ provides a two-dimensional Jordan block,
as long as $\mathsf{S}$ has a single null eigenvalue
(\textit{i.e.}, unless $a=b=0$).

\subsection{Secular equation for $4\times 4$ stability matrix}
\label{app:4*4-eq}

The $D=2$ case, which gives $4\times 4$ stability matrix,
also supplies several instructive examples.
The matrices of Eq.~(\ref{eq:sym-AB}) are expressed as
\be A=\begin{pmatrix} a & b\,e^{i\theta_b} \\
 b\,e^{-i\theta_b} & f \end{pmatrix}\,,
\quad B=\begin{pmatrix} c\,e^{i\theta_c} & d\,e^{i\theta_d} \\
 d\,e^{i\theta_d} & h\,e^{i\theta_h} \end{pmatrix}\,.\quad
 (a,b,c,d,f,h,\theta_b,\theta_c,\theta_d,\theta_h\in\mathbf{R})
\label{eq:4*4}\ee
The l.h.s. of the secular equation (Eq.~(\ref{eq:secular-lhs})) becomes
\be\begin{split}
 \det(\mathsf{S}-\omega\,\mathsf{N})
 &= \omega^4-\big[(a^2-c^2)+(f^2-h^2)+2(b^2-d^2)\big]\,\omega^2\\
 &~+(a^2-c^2)\,(f^2-h^2)+(b^2-d^2)^2-2\,(b^2+d^2)\,a\,f\\
 &~+4\,b\,d\,\big[a\,h\,\cos(\theta_b-\theta_d+\theta_h)
 +c\,f\,\cos(\theta_b-\theta_c+\theta_d)\big]\\
 &~-2\,c\,h\,\big[b^2\,\cos(2\theta_b-\theta_c+\theta_h)
 +d^2\,\cos(2\theta_d-\theta_c-\theta_h)\big]\,.
\end{split}\label{eq:4*4secular-lhs}\ee

\subsection{Quartet solution in $D=2$ model}
\label{app:4*4-quartet}

The $D=2$ model of the previous subsection
provides a simplest example of quartet solutions.
In order that solutions for $\omega$ would form a quartet,
it is necessary and sufficient
for the secular equation with Eq.~(\ref{eq:4*4secular-lhs})
to give a pair of complex conjugates for $\omega^2$, not real $\omega^2$.
Hence the following condition for quartet solutions is obtained
from Eq.~(\ref{eq:4*4secular-lhs}),
\be\begin{split}
 &\frac{1}{4}(a^2-c^2-f^2+h^2)^2 + 2\,(b^2+d^2)\,a\,f
 + (a^2-c^2+f^2-h^2)\,(b^2-d^2) \\
 &\quad < 4\,b\,d\,\big[a\,h\,\cos(\theta_b-\theta_d+\theta_h)
 +c\,f\,\cos(\theta_b-\theta_c+\theta_d)\big]\\
 &\qquad -2\,c\,h\,\big[b^2\,\cos(2\theta_b-\theta_c+\theta_h)
 +d^2\,\cos(2\theta_d-\theta_c-\theta_h)\big]\,.
\end{split}\label{eq:4*4-quartet}\ee
Since this requires $\det\,\mathsf{S}>0$,
it is obvious that two of the four eigenvalues
of the stability matrix $\mathsf{S}$
are positive and the other two are negative.

An immediate example of quartet solutions is obtained by setting $b=c=h=0$,
with $a\ne f$ and $|a+f|<2d$ imposed; \textit{i.e.},
${\displaystyle A=\begin{pmatrix} a & 0 \\ 0 & f \end{pmatrix}}$,
${\displaystyle B=\begin{pmatrix} 0 & d\,e^{i\theta_d} \\ d\,e^{i\theta_d} & 0
\end{pmatrix}}$.
Then the matrix $\mathsf{N\,S}$ has eigenvalues
$[\pm(a-f)\pm i\sqrt{4d^2-(a+f)^2}\,]/2$.
In this particular case,
both the positive and negative eigenvalues of $\mathsf{S}$
have two-fold degeneracy.

\subsection{Jordan blocks for pure-imaginary eigenvalues
in $D=2$ model}
\label{app:4*4-pure-imag}

Pure-imaginary solutions may produce Jordan blocks, if degenerate.
An example is given by setting $a=c$, $f=h$, $b=0$, $\theta_c=\theta_h=0$
and $\theta_d=\pi/2$ in Eq.~(\ref{eq:4*4}); \textit{i.e.},
${\displaystyle A=\begin{pmatrix} a & 0 \\ 0 & f \end{pmatrix}}$,
${\displaystyle B=\begin{pmatrix} a & id \\ id & f \end{pmatrix}}$.
The matrix $\mathsf{N\,S}$ has eigenvalues $\pm id$
with the associating eigenvectors
${\displaystyle\vect{x}_\pm
\propto\begin{pmatrix}1\\ \mp 1\\ -1\\ \pm 1\end{pmatrix}}$,
each forming a 2-dimensional Jordan block
(if $a$, $d$ and $f$ are all different and non-zero).

\subsection{4-dimensional Jordan block for null eigenvalue
in $D=2$ model}
\label{app:4*4-null}

A 4-dimensional Jordan block comes out
for the null eigenvalue in the $D=2$ model
by setting $a=c$, $f=h$, $b=d$, $\theta_b=\theta_d=\theta_h=0$
and $\theta_c=\pi/2$ in Eq.~(\ref{eq:4*4})
(if $a$, $b$ and $f$ are all different and non-zero); \textit{i.e.},
${\displaystyle A=\begin{pmatrix} a & b \\ b & f \end{pmatrix}}$,
${\displaystyle B=\begin{pmatrix} ia & b \\ b & f \end{pmatrix}}$.
This provides an example of Jordan blocks with dimension higher than two.
There is only a single eigenvector,
${\displaystyle\vect{x}
\propto\begin{pmatrix}0\\ 1\\ 0\\ -1\end{pmatrix}}$.



\begin{thebibliography}{99}
\bibitem{ref:RS80} P. Ring and P. Schuck,
 \textit{The Nuclear Many-Body Problem} (Springer-Verlag, 1980).
\bibitem{ref:PY89} R.G. Parr and W. Yang,
 \textit{Density Functional Theory of Atoms and Molecules}
 (Oxford University Press, 1989).
\bibitem{ref:ED11} E. Engel and R.M. Dreizler,
 \textit{Density Functional Theory} (Springer-Verlag, 2011).
\bibitem{ref:PB52} D. Pines and D. Bohm, Phys. Rev. \textbf{85}, 338 (1952).
\bibitem{ref:Thou60} D.J. Thouless, Nucl. Phys. \textbf{21}, 225 (1960).
\bibitem{ref:SM84} Y.R. Shimizu and K. Matsuyanagi,
 Prog. Theor. Phys. \textbf{72}, 1017 (1984).
\bibitem{ref:HNMM} N. Hinohara, T. Nakatsukasa, M. Matsuo and K. Matsuyanagi,
 Prog. Theor. Phys. \textbf{117}, 451 (2007).
\bibitem{ref:Thou61} D.J. Thouless, Nucl. Phys. \textbf{22}, 78 (1961).
\bibitem{ref:TV62} D.J. Thouless and J.G. Valatin,
 Nucl. Phys. \textbf{31}, 221 (1962).
\bibitem{ref:Don99} P. Donati, T. D\o ssing, Y.R. Shimizu, P.F. Bortignon
 and R.A. Broglia, Nucl. Phys. A \textbf{653}, 27 (1999).
\bibitem{ref:SB75} S. Shlomo and G. Bertsch,
 Nucl. Phys. A \textbf{243}, 507 (1975).
\bibitem{ref:Neergard} K. Neerg\aa rd, private communication.
\bibitem{ref:UR71} N. Ullah and D.J. Rowe,
 Nucl. Phys. A \textbf{163}, 257 (1971).
\bibitem{ref:Row70} D.J. Rowe, \textit{Nuclear Collective Motion}
 (Methuen, 1970).

\end{thebibliography}
%

\vfill\pagebreak

\end{document}